\documentclass[a4paper,UKenglish,cleveref, autoref, thm-restate]{lipics-v2019-sample-article}

\title{How Hard is Safe Bribery?\footnote{Accepted for oral presentation at AAMAS 2022}} 

\titlerunning{How Hard is Safe Bribery?} 

\author{Neel Karia}{Microsoft Research, Bengaluru, India}{t-neelkaria@microsoft.com, neel2karia@gmail.com}{}{}

\author{Faraaz Mallick}{IIT Kharagpur, Kharagpur, India}{faraazrm@iitkgp.ac.in}{}{}

\author{Palash Dey}{IIT Kharagpur, Kharagpur, India}{palash.dey@cse.iitkgp.ac.in}{}{}

\authorrunning{N. Karia,  F. Mallick and P. Dey} 

\Copyright{Open Access} 



\nolinenumbers 

\hideLIPIcs  

\EventEditors{}
\EventNoEds{0}
\EventLongTitle{arXiv}
\EventShortTitle{arXiv}
\EventAcronym{arXiv}
\EventYear{2022}
\EventDate{2022}
\EventLocation{}
\EventLogo{}
\SeriesVolume{}
\ArticleNo{0}

\pdfoutput=1
\usepackage{cleveref}
\usepackage{makecell}
\usepackage{multirow}
\usepackage{mathtools}
\usepackage{eulervm}
\usepackage{wrapfig}
\usepackage{lipsum}
\usepackage{float}
\usepackage[ruled,vlined]{algorithm2e}
\usepackage{placeins}
\usepackage{threeparttable}
\usepackage{hhline}
\usepackage{multirow}

\usepackage{url}

\usepackage{xspace}
\usepackage{amsmath}

\newcommand{\NPC}{\ensuremath{\mathsf{NP}}-complete\xspace}
\newcommand{\NPH}{\ensuremath{\mathsf{NP}}-hard\xspace}

\newcommand{\Pb}{\ensuremath{\mathsf{P}}\xspace}

\newcommand{\coNP}{\ensuremath{\mathsf{co}}-\ensuremath{\mathsf{NP}}\xspace}

\newcommand{\coNPC}{\ensuremath{\mathsf{co}}-\ensuremath{\mathsf{NP}}-complete\xspace}
\newcommand{\WOH}{\ensuremath{\mathsf{W[1]}}\text{-hard}\xspace}
\newcommand{\coWOH}{\ensuremath{\mathsf{co}}-\ensuremath{\mathsf{W[1]}}-hard\xspace}
\newcommand{\WKH}{\ensuremath{\mathsf{W[k]}}-hard\xspace}
\newcommand{\WTH}{\ensuremath{\mathsf{W[2]}}\text{-hard}\xspace}
\newcommand{\WO}{\ensuremath{\mathsf{W[1]}}\xspace}
\newcommand{\WT}{\ensuremath{\mathsf{W[2]}}\xspace}
\newcommand{\WP}{\ensuremath{\mathsf{W[P]}}\xspace}
\newcommand{\XP}{\ensuremath{\mathsf{XP}}\xspace}
\newcommand{\FPT}{\ensuremath{\mathsf{FPT}}\xspace}

\newcommand{\el}{\ensuremath{\ell}\xspace}

\newcommand{\suc}{\ensuremath{\succ}\xspace}

\newcommand{\YES}{{\sc{yes}}\xspace}
\newcommand{\NO}{{\sc{no}}\xspace}
\newcommand{\no}{{\sc{no}}\xspace}
\newcommand{\yes}{{\sc{yes}}\xspace}

\newcommand{\IF}{{\sc{If}}\xspace}
\newcommand{\ELSE}{{\sc{Else}}\xspace}

\newcommand{\IS}{\textsc{Is Safe}\xspace}
\newcommand{\DIS}{\textsc{\$Bribery Is Safe}\xspace}
\newcommand{\SDB}{\textsc{Safe \$Bribery}\xspace}
\newcommand{\SBIS}{\textsc{Shift Bribery Is Safe}\xspace}
\newcommand{\SSB}{\textsc{Safe Shift Bribery}\xspace}
\newcommand{\DB}{\textsc{\$Bribery}\xspace}

\newcommand{\SHB}{\textsc{Shift Bribery}\xspace}
\newcommand{\SWB}{\textsc{Swap Bribery}\xspace}
\newcommand{\XTC}{\textsc{Exact Cover by 3-Sets}\xspace}
\newcommand{\XTFC}{\textsc{Exact Cover by (3,4)-Sets}\xspace}
\newcommand{\MCI}{\textsc{Multicoloured Independent Set}\xspace}
\newcommand{\MCC}{\textsc{Multicoloured Clique}\xspace}

\renewcommand{\AA}{\ensuremath{\mathcal A}\xspace}

\newcommand{\CC}{\ensuremath{\mathcal C}\xspace}
\newcommand{\DD}{\ensuremath{\mathcal D}\xspace}
\newcommand{\EE}{\ensuremath{\mathcal E}\xspace}

\newcommand{\GG}{\ensuremath{\mathcal G}\xspace}

\newcommand{\LL}{\ensuremath{\mathcal L}\xspace}

\newcommand{\OO}{\ensuremath{\mathcal O}\xspace}
\newcommand{\PP}{\ensuremath{\mathcal P}\xspace}
\newcommand{\QQ}{\ensuremath{\mathcal Q}\xspace}
\newcommand{\RR}{\ensuremath{\mathcal R}\xspace}
\renewcommand{\SS}{\ensuremath{\mathcal S}\xspace}

\newcommand{\UU}{\ensuremath{\mathcal U}\xspace}
\newcommand{\VV}{\ensuremath{\mathcal V}\xspace}

\newcommand{\XX}{\ensuremath{\mathcal X}\xspace}

\newcommand{\bbb}{\ensuremath{\mathfrak b}\xspace}

\newcommand{\sss}{\ensuremath{\mathfrak s}\xspace}

\newcommand{\vvv}{\ensuremath{\mathfrak v}\xspace}

\usepackage{nicefrac}

\newcommand{\nfrac}{\nicefrac}

\newcommand{\eps}{\varepsilon}
\renewcommand{\epsilon}{\eps}

\newcommand{\ignore}[1]{}

\renewcommand{\leq}{\leqslant}
\renewcommand{\geq}{\geqslant}
\renewcommand{\ge}{\geqslant}

\usepackage{cleveref}
\usepackage {enumerate}

\crefname{property}{Property}{Properties}
\crefname{example}{Example}{Examples}
\crefname{theorem}{Theorem}{Theorems}
\crefname{observation}{Observation}{Observations}
\crefname{lemma}{Lemma}{Lemmas}
\crefname{corollary}{Corollary}{Corollaries}
\crefname{proposition}{Proposition}{Propositions}
\crefname{definition}{Definition}{Definitions}
\crefname{claim}{Claim}{Claims}
\crefname{reductionrule}{Reduction rule}{Reduction rules}
\crefname{ineq}{inequality}{Inequalities}

\begin{document}

\maketitle

\begin{abstract}
Bribery in an election is one of the well-studied control problems in computational social choice. In this paper, we propose and study the safe bribery problem. Here the goal of the briber is to ask the bribed voters to vote in such a way that the briber never prefers the original winner (of the unbribed election) more than the new winner, even if the bribed voters do not fully follow the briber's advice. Indeed, in many applications of bribery, campaigning for example, the briber often has limited control on whether the bribed voters eventually follow her recommendation and thus it is conceivable that the bribed voters can either partially or fully ignore the briber's recommendation. We provide a comprehensive complexity theoretic landscape of the safe bribery problem for many common voting rules in this paper.
\end{abstract}

\section{Introduction}
\label{sec:typesetting-summary}
Voting has always served as a fundamental tool for aggregating varied preferences in many applications in real-life and artificial intelligence~\cite{PennockHG00,jackson2008consensus}. In a typical voting setting, we have a set of candidates, a set of voters each having a preference over the set of candidates, and a voting rule which decides a winner based on the preferences of the voters. Any such voting scenario is susceptible to various kinds of control attacks -- voters or candidates or some other third party may influence the outcome of the election in their way through some unfair means~\cite{DBLP:reference/choice/FaliszewskiR16}. One of the most well-studied attacks of such type is {\em bribery,} where an external agent, called a briber, offers monetary rewards to some voters to vote as the briber suggests so that the favourite candidate of the briber wins the election~\cite{faliszewski2006complexity,faliszewski2009hard,faliszewski2009llull}. This bribery problem not only models the monetary bribing but also other situations like campaigning in an election where the monetary reward corresponds to the amount of time and energy one needs to spend to campaign for some candidate. Depending on how the briber needs to pay the voters to change their votes, various models have been studied. In the \DB problem, each voter has a fixed cost that the briber needs to pay to make the voter cast the vote of the briber's choice~\cite{faliszewski2006complexity,faliszewski2009hard}. In \SWB, the briber has to pay for each swap of consecutive candidates in a voter's preference and the price also depends on the pair of candidates being swapped~\cite{elkind2009swap}. In \SHB, the briber can only shift her favourite candidate by some positions and the cost depends on the number of positions shifted~\cite{bredereck2016complexity,kaczmarczyk2016algorithms,bredereck2014prices,maushagen2018complexity}.

To the best of our knowledge, all the existing work on bribery assumes that the bribed voters cast the exact same vote that the briber has asked them. Although this may be a reasonable assumption for some applications, in many other applications of bribery, say campaigning, the briber can never be sure that the voters will eventually cast the vote of the briber's choice. Indeed, it may very well be the case that some subset of voters follows the briber's recommendation exactly, some other subset of voters follows them partially, and the remaining voters complete ignore the briber. Moreover, the briber may not have any knowledge of these sets of voters. Now the situation will be worse for the briber if she prefers the original winner more than the winner of the resulting election, where not all bribed voters cast the votes of the briber's choice. For example, let us consider a plurality election where $10$ voters vote for candidate $a$, $8$ voters vote for candidate $b$, and $4$ voters vote for candidate $c$. Assume that the briber's preference is $c\suc a\suc b$ and the briber bribes $6$ voters voting originally for $a$ to vote for $c$. However, only $3$ of the $6$ bribed voters eventually follow the briber's recommendation, while the other $3$ bribed voters simply ignore the briber. In the resulting election, candidates $a$ and $c$ receive $7$ votes each whereas the candidate $b$ wins the election with $8$ votes. We observe that the briber would prefer the original winner $a$ over the new winner $b$. To model this kind of applications more suitably, we propose the {\em safe bribery} problem.

In the safe bribery problem, the briber has a preference $\suc_B$ which is a complete order over the candidates. Given a preference profile of a set of voters, a cost function for each voter, a favourite candidate $c$ of the briber, and a budget of the briber, the goal of the safe bribery problem is to compute if there exists a subset of voters who can be bribed in such a way that (i) if all the bribed voters follow the briber's recommendation, then $c$ wins, and (ii) for every subset of bribed voters who follows the briber's recommendation exactly, every subset of bribed voters who follows the briber's recommendation partially, and the other bribed voters ignoring the briber's recommendation, the winner of the election is not less preferred than the winner of the ``unbribed'' election according to $\suc_B$. We study the safe version of the \DB and \SHB problems in this paper. We also study the computational problem, called {\em is safe,} of deciding if a given bribed profile and a given unbribed profile is safe for the briber with respect to a preference of the briber. We refer the reader to \Cref{sec:prob} for the formal definitions of the four problems. \Cref{sec:algo} covers some algorithms for the polynomial-time results, \Cref{sec:hard} covers the hardness results, and \Cref{sec:param} deals with the parameterized complexity results.

\subsection{Contribution} We provide a comprehensive complexity theoretic landscape of the \DIS, \SDB, \SBIS, and \SSB problems. We summarise our main results in \Cref{tab:Res}. Other than these results, we also show that all the four problems are polynomial time solvable for every anonymous and efficient voting rule when we have a constant number of candidates [\Cref{theorem const is safe,theorem const safe}]. We also look at safety in \SHB from a parameterized hardness perspective, and summarise our results in \Cref{tab:ResParamIsSafe,tab:ResParamSafe}.

\begin{table}[!htbp]
	\centering
	\begin{tabular}{ |p{1.7cm}|p{2.7cm}|p{2.3cm}|p{2.4cm}|p{2.3cm}| }
		\hline
		Voting Rule & \DIS & \SDB & \SBIS & \SSB\\
		\hline
		
		Plurality & \Pb (\Cref{plu is safe}) & \Pb (\Cref{plu safe}) & \Pb (\Cref{plu is safe}) & \Pb (\Cref{plu safe}) \\
		\hline
		$k$-approval & \coNPC (\Cref{Theorem k-approval dol}) & \NPH (\Cref{Coro Obs 1}) & \Pb (\Cref{k-approval shift}) & \Pb (\Cref{k-approval shift})\\
		\hline
		Veto & \Pb (\Cref{Theorem veto shift}) & \Pb (\Cref{theorem veto dol}) & \Pb (\Cref{Theorem veto shift}) & \Pb (\Cref{theorem veto safe shift})\\
		\hline
		$k$-veto & \coNPC (\Cref{theorem k-veto}) & \NPH (\Cref{Coro Obs 1}) &  \Pb (\Cref{theorem k-veto}) & \Pb (\Cref{theorem k-veto})\\
		\hline
		Borda & \coNPC (\Cref{theorem borda dol}) & \NPH (\Cref{Coro Obs 1}) & \coNPC (\Cref{theorem borda shift}) & \NPH (\Cref{Coro Obs 1})\\
		\hline
		S. Bucklin & \coNPC (\Cref{theorem bucklin dol})  & \NPH (\Cref{Coro Obs 1}) & \Pb (\Cref{theorem bucklin shift}) & \Pb (\Cref{theorem bucklin safe shift}) \\
		\hline
		Copeland & \coNPC (\Cref{theorem condorcet dol})  & \NPH (\Cref{Coro Obs 1}) & \coNPC (\Cref{Theorem Copeland}) & \NPH (\Cref{Coro Obs 1})\\
		\hline
		Maximin & \coNPC (\Cref{theorem condorcet dol}) & \NPH (\Cref{Coro Obs 1}) & \coNPC (\Cref{theorem maximin shift}) & \NPH (\Cref{Coro Obs 1}) \\
		\hline
	\end{tabular}
	\caption{Complexity Results for Safe Bribery}
	\label{tab:Res}
\end{table}

\begin{table*}[!htbp]
	\centering
	\begin{tabular}{ |c|c|c|c| }
		\hline
		
		Voting Rule & \#shifts & \#bribed voters & \#candidates\\
		\hline
		Borda & \FPT (\Cref{param is safe shifts}) & \coWOH (\Cref{borda param voters}) & \XP (\Cref{theorem const is safe})\\
		\hline
		Copeland & \FPT (\Cref{param is safe shifts}) & \coWOH (\Cref{copeland param voters}) & \XP (\Cref{theorem const is safe})\\
		\hline
		Maximin & \FPT (\Cref{param is safe shifts}) & ? & \XP (\Cref{theorem const is safe})\\
		\hline
	\end{tabular}
	\caption{Parameterized Complexity Results for \SBIS}
	\label{tab:ResParamIsSafe}
\end{table*}

\begin{table*}[!htbp]
	\centering
	\begin{tabular}{ |p{1.7cm}|p{2.5cm}|p{2.5cm}|p{2cm}|p{2.1cm}|}
		\hline
		
		Voting Rule & \#shifts & \#candidates & \#voters & \#bribed voters OR budget\\
		\hline
		Borda & \XP (\Cref{param safe shifts}) & \XP (\Cref{param safe shifts}) & \WOH (\Cref{Coro 2}) & \WTH (\Cref{Coro 2})\\
		\hline
		Copeland & \WOH (\Cref{Coro 2}) & \XP (\Cref{param safe shifts}) & \WOH (\Cref{Coro 2}) & \WTH (\Cref{Coro 2})\\
		\hline
		Maximin & \XP (\Cref{param safe shifts}) & \XP (\Cref{param safe shifts}) & \WOH (\Cref{Coro 2}) & \WTH (\Cref{Coro 2})\\
		\hline
	\end{tabular}
	\caption{Parameterized Complexity Results for \SSB}
	\label{tab:ResParamSafe}
\end{table*}

\subsection{Related Work}

Faliszewski et al.~\cite{faliszewski2006complexity} propose the first bribery problem where the briber's goal is to change the minimum number of preferences to make some candidates win the election. Then they extend their basic model to more sophisticated models, including \DB~\cite{faliszewski2009hard,faliszewski2009llull}. Elkind et al.~\cite{elkind2009swap} extend this model further and study the \SWB problem (where there is a cost associated with every swap of candidates), and its special case, the \SHB problem. Dey et al.~\cite{frugalDeyMN16} show that the bribery problem remains intractable for many common voting rules for an interesting special case which they call frugal bribery. The bribery problem has also been studied in various other preference models, for example, truncated ballots~\cite{BaumeisterFLR12}, soft constraints~\cite{pini2013bribery}, approval ballots~\cite{schlotter2017campaign}, campaigning in societies~\cite{faliszewski2018opinion}, CP-nets~\cite{dorn2014hardness}, combinatorial domains~\cite{mattei2012bribery}, iterative elections~\cite{maushagen2018complexity}, committee selection~\cite{bredereck2016complexity}, probabilistic lobbying~\cite{erdelyi2009complexity}, local distance restricted bribery~\cite{DBLP:journals/tcs/Dey21, baumeister2019generalized, yang2016hard} etc. Erdelyi et al.~\cite{erdelyi2014bribery} study the bribery problem under voting rule uncertainty. Faliszewski et al.~\cite{faliszewski2014complexity} 
study bribery for the simplified Bucklin and the Fallback voting rules, and in doing this they explore some novel proof techniques. Xia~\cite{xia2012computing}, Kaczmarczyk and Faliszewski~\cite{kaczmarczyk2016algorithms}, as well as Yang ~\cite{yang2020complexity} study the destructive variant of bribery. Dorn and Schlotter~\cite{dorn2012multivariate} and Bredereck et al.~\cite{bredereck2014prices} explore parameterized complexity of various bribery problems. Chen et al.~\cite{chen2018protecting} provide novel mechanisms to protect elections from bribery. Knop et al.~\cite{Knop} provide a unified framework for various control problems. In control, there is addition/deletion of voters or candidates. However, in case of bribery, the votes are being modified. This is a key difference, leading to slightly differing techniques. Although most of the bribery problems are intractable, few of them, \SHB for example, have polynomial time approximation algorithms~\cite{elkind2010approximation,DBLP:conf/aaai/KellerHH18}. Manipulation, a specialization of bribery, is another fundamental attack on election~\cite{DBLP:reference/choice/ConitzerW16}. In the manipulation problem, a set of voters (called manipulators) wants to cast their preferences in such a way that (when tallied with the preferences of other candidates with matching preferences) makes some candidate win the election. Obraztsova and Elkind~\cite{DBLP:conf/aaai/ObraztsovaE12,DBLP:conf/aamas/ObraztsovaE12} initiate the study of optimal manipulation in that context. However, this is different from bribery in the sense that the voters are not incentivized to vote to match their preferences with some other voters, but are in fact incentivized to modify their vote according to the briber's choice.

The concept of safety in electoral control problems has been studied before as well. Slinko and White initiate this line of work by proposing the notion of safety in the context of manipulation and studying the class of social choice functions that are safely manipulable~\cite{slinko2008nondictatorial,slinko2014ever}. Hazon and Elkind~\cite{hazon2010complexity} and Ianovski et al.~\cite{ianovski2011complexity} study computational complexity of safely manipulating popular voting rules.

\section{Preliminaries}
	
An election is a pair $(\CC, \VV)$, where $\CC = \{c_1,\ldots,c_m\}$ is a set of candidates and $\VV = \{v_1 , \ldots, v_n\}$ is a set of voters. If not mentioned otherwise, we use $n$ and $m$ to respectively denote the number of voters and candidates. Each voter $v_i$ has a 
preference order (vote) $\succ_i$, which is a linear order over $\CC$.
We denote the set of all complete orders over $\CC$ by $\LL(\CC)$. We call a list of $n$ preference orders $\{\succ_1, \succ_2, \ldots,\succ_n\} \in \LL(\CC)^n$ an $n$-voter preference profile. We denote the $i^\textsuperscript{th}$ preference order of a preference profile \PP as $\succ_i^\PP$. A mapping $r : \LL(\CC)^n \rightarrow \CC$ is called a resolute voting rule (as we assume the unique-winner model); in case of ties, the winner is decided by a \textit{lexicographic} tie-breaking mechanism $\succ_t$, which is some pre-fixed ordering over the candidates. For a set of candidates $X$, let $\overrightarrow{X}$ be an ordering over them. Then, $\overleftarrow{X}$ denotes the reversed ordering of $\overrightarrow{X}$.

Let $[\el]$ represent the set of positive natural numbers up to $\el$ for any positive integer \el. A voting rule $r$ is called \textit{anonymous} if, for every preference profile $(\succ_i )_{i \in [n]} \in \LL(\CC)^n$ and permutation $\sigma$ of $[n]$, we have $r((\succ_i )_{i \in [n]} ) = r((\succ_{\sigma(i)} )_{i \in [n]} )$. A voting rule is called \textit{efficient} if the winner can be computed in polynomial time of the input length. 

A scoring rule is induced by an $m$-dimensional vector, $(\alpha_1, \ldots, \alpha_m) \in \mathbb{Z}^m$ with $\alpha_1 \geq \alpha_2 \geq \ldots \geq \alpha_m$ and $\alpha_1 > \alpha_m$. A candidate gets a score of $\alpha_i$ from a voter if she is placed at the $i^{th}$ position in the voter's preference order. The score of a candidate from a voter set is the sum of the scores she receives from each of the voters. If $\alpha_i$ is $1$ for $i \in [k]$ and $0$ otherwise, we get the $k$-approval rule. If $\alpha_i$ is $0$ for $i \in [m-k]$ and $-1$ otherwise, we get the $k$-veto rule. 
The scoring rules for score vectors $(1, 0, \ldots , 0)$ and $(0, \ldots , 0, -1)$ are called plurality and veto rules respectively. The scoring rule for score vector $(m - 1, m - 2, \ldots , 1, 0)$ is known as the Borda rule. 

The scores for the other voting rules studied in the paper (apart from scoring rules) are defined as follows. Let $vs(a,b)$ be the difference in the number of votes in which $a$ precedes $b$ and the number of votes in which $a$ succeeds $b$, for $a,b \in \CC$.
The maximin score of a candidate $a$ is $\min_{b \neq a} vs(a, b)$. The candidate with the maximum maximin score (after tie-breaking) is the winner.
Given $\alpha \in [0, 1]$, the Copeland$^\alpha$ score of a candidate $a$ is $|\{b \neq a : vs(a, b) > 0\}| + \alpha \times |\{b \neq a : vs(a, b) = 0\}|$. The candidate with the maximum Copeland$^\alpha$ score (after tie-breaking) is the winner. We will assume $\alpha$ to be zero, if not mentioned separately. A candidate $a$ that has a positive pairwise score $vs(a, b)$ against all other candidates $b \in \CC$ is called a Condorcet winner. Rules which select a Condorcet winner as the winner (whenever it exists) are called Condorcet-consistent rules. Copeland and maximin voting rules are Condorcet-consistent.
The simplified Bucklin score of a given candidate $a$ is the minimum number $k$ such that more than half of the voters rank $a$ in their top $k$ positions. The candidate with the lowest simplified Bucklin score (after tie-breaking) is the winner and her score is called the Bucklin winning round.

We use $s(c)$ to denote the total score that a candidate $c \in \CC$ gets in an election. Similarly, if $Y \subseteq \CC$, we use $s(Y)$ to refer to the score of each candidate from $Y$; e.g. $s(Y)=5$ means that each candidate from $Y$ has a score of $5$. The voting rule under consideration will be clear from the context. Note that we assume that the briber is aware of the votes cast by each voter. 
	
\subsection{Parameterized Complexity} 
A parameterized problem $\Pi$ is a 
subset of $\Gamma^{*}\times
\mathbb{N}$, where $\Gamma$ is a finite alphabet. A central notion is \emph{fixed parameter 
	tractability} (FPT) which means, for a 
given instance $(x,k)$, solvability in time $f(k) \cdot p(|x|)$, 
where $f$ is an arbitrary function of $k$ and 
$p$ is a polynomial in the input size $|x|$. There exists a hierarchy of complexity classes above FPT, and showing that a parameterized problem is hard for one of these classes is considered
evidence that the problem is unlikely to be fixed-parameter tractable. The main classes in this hierarchy are: $ \FPT  \subseteq \WO \subseteq \WT \subseteq \cdots \subseteq \WP \subseteq \XP.$ We now define the notion of parameterized reduction~\cite{cygan2015parameterized}.

\begin{definition}
	Let $A,B$ be parameterized problems.  We say that $A$ is {\bf \em fpt-reducible} to $B$ if there exist functions 
	$f,g:\mathbb{N}\rightarrow \mathbb{N}$, a constant $\alpha \in \mathbb{N}$ and 
	an algorithm $\Phi$ which transforms an instance $(x,k)$ of $A$ into an instance $(x',g(k))$ of $B$ 
	in time $f(k) |x|^{\alpha}$
	so that $(x,k) \in A$ if and only if $(x',g(k)) \in B$.
\end{definition}

\subsection{Problem Definition}\label{sec:prob}
	
Let $r$ be any voting rule. We first define when a bribed profile is ``safe'' in the context of \DB. Let $\succ_B$ be the preference of the briber. Intuitively speaking, we say that a bribed profile is safe if no candidate preferred less than the original winner (in $\succ_B$), wins the election when a subset of bribed voters does not follow the briber's suggestion but casts their original votes. Formally, we define the notion of safety for \DB as follows.

\begin{definition}
	(Safety for \DB): Given a voting rule $r$, a set $\CC$ of $m$ candidates, a set $\VV$ of $\ n$ voters, an $n$-voter profile $\PP = (\succ_i^{\PP})_{i \in [n]} \in {\LL(\CC)}^n$, the most preferred candidate $c \in \CC$, a preference order $\succ_B \in \LL(\CC)$ of the briber, and another profile $\QQ = (\succ_i^{\QQ})_{i \in [n]} \in {\LL(\CC)}^n$ (representing a bribed profile), we say that $\QQ$ is safe for $\PP$ with respect to $\succ_B$ if the following conditions hold. Let us define $w = r(\PP)$ (where $c \succ_B w$), and $\VV_b = \{v_i \ \mid \ v_i \in \VV, \ \succ_i^{\PP} \neq \succ_i^{\QQ}\}$. We call the voters in $\VV_b$ the bribed voters, and define $\VV_u = \VV \setminus \VV_b$.
	
	\begin{itemize}
		\item \textbf{[Safety]} For every subset $\VV_b' \subseteq \VV_b$, if we have $x = r((\succ_i^{\QQ})_{v_i \in \VV_b'} ,(\succ_i^{\PP})_{v_i \in \VV \setminus \VV_b'})$, then we have $x \succ_B w$ or $x = w$.
		\item \textbf{[Success]} We have $c = r((\succ_i^{\QQ})_{v_i \in \VV_b} ,(\suc_i^{\PP})_{v_i \in \VV_u})$.
	\end{itemize}
\end{definition} 

We now define the computation problem of finding if a bribed profile is safe.

\begin{definition} (\DIS): Given a voting rule $r$, a set $\CC$ of $m$ candidates, an $n$-voter profile $\PP = (\succ_i^{\PP})_{i \in [n]} \in \LL(\CC)^n$ over $\CC$,the most preferred candidate $c \in \CC$, a preference order $\succ_B \in \LL(\CC)$ of the briber, and another profile $\QQ = (\succ_i^{\QQ})_{i\in[n]} \in \LL(\CC)^n$ (representing the bribed profile), compute if $\QQ$ is safe for $\PP$ with respect to $\succ_B$ for $r$. We denote any arbitrary instance of \DIS by $(\CC, \PP, c, \succ_B, \QQ)$.
\end{definition}

We next define the computational problem of safely bribing the voters in an election.

\begin{definition} (\SDB): Given a voting rule $r$, a set $\CC$ of $m$ candidates, a set \VV of $n$ voters, an $n$-voter profile $\PP = (\succ_i^{\PP})_{i \in [n]} \in
	\LL(\CC)^n$ corresponding to \VV, the most preferred candidate $c \in \CC$, a preference $\succ_B \in \LL(\CC)$ of the briber, a family $\Pi = (\pi_i)_{i \in [n]} \in \mathbb{N}^n$ of cost functions corresponding to every voter, and a budget $b \in \mathbb{R}$, compute
	if there exists a set $\VV_b \subseteq \VV$ of voters along with a profile $\QQ = ((\succ_i^{\QQ})_{v_i\in \VV_b},(\succ_i^{\PP})_{v_i \in \VV \setminus \VV_b}) \in \LL(\CC)^n$ such that the following conditions hold:
	(1) $\Sigma_{v_i \in \VV_b} \pi_i \leq b$,
	(2) The profile $\QQ$ is safe for $\PP$ with respect to $\succ_B$ for $r$. An arbitrary instance of \SDB is denoted by $(\CC, \PP, c, \succ_B, \Pi, b)$.
\end{definition}

We next define the concept of safety for \SHB. Here, the concept of safety is more fine-grained. Suppose a voter is bribed to shift the favourite candidate $c$ of the briber by $t$ positions. However, it is possible that the bribed voter follows the suggestion of the briber only partially and shifts $c$ to left by a lesser number of positions than $t$. Hence, the briber needs to bribe in this case in such a way that no unfavourable candidate for the briber (compared to the original winner) wins the election even if any subset of voters follow the briber's suggestion only partially. Formally, it is defined as follows.

\begin{definition} (Safety for \SHB): Given a voting rule $r$, a set $\CC$ of $m$ candidates, a set of $\VV$ of $n$ voters, an $n$-voter profile $\PP = (\succ_i^\PP)_{i \in [n]} \in \LL(\CC)^n$ corresponding to \VV, the most preferred candidate $c \in \CC$, a preference order $ \succ_B \ \in \LL(\CC)$ of the briber, and a shift vector $\sss = (s_1, \dots , s_n) \in \mathbb{N}_0^n$, we say that the shift vector $\sss$ is safe for $\PP$ with respect to $\succ_B$ if the following conditions hold. Let us define $w = r(\PP)$ (where $c \succ_B w$) and $\QQ = \{\suc_i^\QQ \ \mid \ i \in [n]$, $\suc_i^\QQ$ is obtained from $\suc_i^\PP$ by shifting $c$ to the left by $s_i$ positions$\}$
	
	\begin{itemize}
		\item \textbf{[Safety]} For every shift vector $\sss' = (s'_1, \dots , s'_n)$ with $s'_i \leq s_i, \forall i \in [n]$, if we have $\QQ' = \{ \suc_i^\QQ \ \mid \ i \in [n]$, $\suc_i^\QQ$ is obtained from $\suc_i^\PP$ by shifting $c$ to the left by $s'_i$ positions\} and $x = r(\QQ')$, then we have $x \succ_B w$ or $x = w$.
		
		\item \textbf{[Success]} We have $c = r(\QQ)$.
	\end{itemize}
\end{definition}

We next define the problem of computing if a \SHB is safe.

\begin{definition} (\SBIS): Given a voting rule $r$, a set $\CC$ of $m$ candidates, an $n$-voter profile $\PP = (\succ_i^{\PP})_{i \in [n]} \in \LL(\CC)^n$ over $\CC$, the most preferred candidate $c \in \CC$, a preference order $\succ_B \in \LL(\CC)$ of the briber, and a shift
	vector $\sss = (s_1, \dots , s_n) \in \mathbb{N}_0^n$ (corresponding to a bribing strategy), compute if $\sss$ is safe for $\PP$ with respect to $\succ_B$ for $r$. We denote any arbitrary instance of \SBIS by $(\CC, \PP, c, \succ_B, \sss)$. Sometimes it will be convenient to define an instance of \SBIS as $(\CC,\PP,c,\succ_B,\QQ)$ where \QQ is obtained by applying $\sss$ to \PP.
\end{definition}

We next define the computational problem of bribing the voters in an election, safely and according to the rules of \SHB.

\begin{definition} (\SSB): Given a voting rule $r$, a set $\CC$ of $m$ candidates, a set $\VV$ of $n$ voters, an $n$-voter profile $\PP = (\succ_i^{\PP})_{i \in [n]} \in \LL(\CC)^n$ over $\CC$, the most preferred candidate $c \in \CC$, a preference $\succ_B \in \LL(\CC)$ of the briber, a family $\Pi = (\pi_i : [m - 1] \rightarrow \mathbb{N})_{i \in [n]}$ of cost functions (where $\pi_i(0) = 0, \forall i \in [n]$) corresponding to every voter, and a budget $b \in \mathbb{R}$, compute if there exists a shift vector $\sss = (s_1, \dots, s_n) \in \mathbb{N}_0^n$ such that:
	(1) $\Sigma_{v_i \in \VV} \pi_i(s_i) \leq b$
	(2) The shift vector $\sss$ is safe for $\PP$ with respect to $\succ_B$. An arbitrary instance of \SSB is denoted by $(\CC, \PP, c, \succ_B, \Pi, b)$.
\end{definition}

\section{Results}

We now present our results. Given an election with its winner $w$ and the preference $\suc_B$ of the briber, we define a set $G$ of ``good candidates'' as $\{a\in\CC: a\suc_B w\}\cup\{w\}$ and the set $B$ of ``bad candidates'' as $\{a\in\CC: w\suc_B a\}$. For \IS, let $\VV_b$ be the set of bribed voters, and $\VV_u$ the rest of the voters.

We begin with showing a connection of \SSB and \SDB with the classical problems \SHB and \DB respectively.

\begin{observation}\label{Obs 1}
	There is a polynomial-time many-to-one reduction from \DB to \SDB and from \SHB to \SSB.
\end{observation}

	Let $w$ be the winner of the unbribed election. To reduce any instance of \DB (\SHB respectively) to \SDB (\SSB respectively), we set the preference $\succ_B$ of the briber to be $c \succ \ldots \succ w$ and keep everything else the same. The reduction is clearly correct and runs in polynomial time.

Many hardness results of \SDB and \SSB are obtained as useful corollaries.

\begin{corollary}\label{Coro Obs 1}
	If \DB (\SHB respectively) is \NPC for any anonymous and efficient voting rule, \SDB (\SSB respectively) is \NPH for that voting rule. 
\end{corollary}

Since \DB is \NPC for $k$-approval, $k$-veto, Borda, simplified Bucklin, Copeland  and maximin  \cite{faliszewski2006complexity, brelsford2008approximability, faliszewski2015complexity, faliszewski2009llull}, \SDB is \NPH for these voting rules. Similarly, since \SHB is \NPC for Borda, Copeland and maximin (using results from \cite{edithswap}), \SSB is \NPH for these voting rules.

\begin{corollary}\label{Coro 2}
If \SHB is \WKH for any anonymous and efficient voting rule, when parameterized by any parameter, \SSB is also \WKH for that voting rule, when parameterized by the same parameter.
\end{corollary}

This is because the reduction in \Cref{Obs 1} is also an fpt-reduction, since the instances of the \SHB and \SSB problems are the same (preserving the parameter), with the only difference being the addition of $\succ_B$ to the \SSB instance, which can be done in polynomial time.
Since \SHB is \WTH for Borda, Copeland and maximin when parameterized by the number of bribed voters or the budget~\cite{bredereck2014prices, bredereck2016complexity}, \SSB is also \WTH for these voting rules when parameterized by either of these two parameters.
Since \SHB is \WOH for Borda, Copeland, and maximin when parameterized by the number of voters, \SSB is also \WOH for these three voting rules with number of bribed voters as the parameter~\cite{bredereck2014prices}. By a similar reduction, \SSB for Copeland is \WOH with respect to the number of shifts, given that \SHB for Copeland is \WOH with respect to the number of shifts~\cite{bredereck2014prices}.

\subsection{Algorithmic Results}\label{sec:algo}

We present here our algorithmic results. We show that all our four problems are polynomial-time solvable for every anonymous and efficient voting rule when we have a constant number of candidates. Our results on the scoring rules follow via an algorithm that uses min-cost flow problem as a crucial subroutine, in a similar manner as in \cite{faliszewski2008nonuniform}.

\begin{theorem}\label{theorem const is safe}
	When we have a constant number of candidates, both \DIS and \SBIS are in \Pb for every anonymous and efficient voting rule. That is, both \DIS and \SBIS belong to \XP with respect to the number of candidates as the parameter.
\end{theorem}
\begin{proof}
	Let $(\CC, \PP, c, \succ_B , \QQ)$ and $(\CC, \PP, c, \succ_B , \sss)$ be any instances of \DIS
	and \SBIS respectively. If $r(\QQ)$ is not $c$, we output \no since the bribery is not successful in this case. The number of possible anonymous preference profiles is $\binom{m!+n-1}{m!-1} = \OO\left((n+m)^m\right) = \OO\left(n^{\OO(1)}\right)$ when we have $m=\OO(1)$. Let \RR be any anonymous preference profile such that $r(\PP)\suc_B r(\RR)$; if no such \RR exists, we output \yes. We now construct a flow network $\GG_\RR (V, E, W)$ to verify if one can obtain the profile \RR from \PP using some bribed voters ignoring the briber's suggestion fully (or partially for shift safe bribery).
	\begin{align*}
	    V &= \{a_i , b_i \mid i \in [n]\} \cup \{s, t\}\\
	    E &= \{(s, a_i ) \mid i \in [n]\} \cup \{(b_i, t) \mid i \in [n]\} \cup F
	\end{align*}
	We now describe the set $F$ of edges. There is an edge $(a_i , b_j)$ in $F$ if for \DIS, $\succ_i^\QQ = \succ_j^\RR or \succ_i^\PP = \succ_j^\RR$ and for \SBIS if $\succ_j^\RR$ can be obtained from $\succ_i^\PP$ by moving $c$ left by at most $s_i$ positions. We finally define the capacity of every edge to be $1$. It is easy to check that one can obtain the profile \RR from \PP considering some bribed voters who ignore the briber's suggestion fully (or partially for shift safe bribery) if and only if there is an $s-t$ flow of value $n$. We output \yes if there is no \RR such that there exists an $s-t$ flow of value $n$. Otherwise, we output \no. Since maximum $s-t$ flow can be computed in polynomial time, our algorithm runs in $\OO\left((n+m)^m\right)\text{poly}(m,n)$ time.\end{proof}

\begin{theorem}\label{theorem const safe}
When we have a constant number of candidates, both \SDB and \SSB are in \Pb for every anonymous and efficient voting rule. That is, both \SDB and \SSB belong to \XP with respect to the number of candidates as the parameter.
\end{theorem}

\begin{proof}
	Let $(\CC,\PP,c,\succ_B,\Pi,b)$ be an instance of \SDB (or \SSB). 
	Let $\QQ=(\succ_i^{\QQ})_{i\in[n]}$ be any anonymous preference profile such that $r(\QQ)=c$; if no such \RR exists, we output \NO. We now construct a minimum cost flow network $\GG_\QQ(V,E,C,W)$ (where $C$ represents the edge costs) to compute the minimum cost of obtaining \QQ as bribed profile from \PP.
	\begin{align*}
	    V &= \{a_i , b_i \mid i \in [n]\} \cup \{s, t\}\\
	    E &= \{(s, a_i ) \mid i \in [n]\} \cup \{(b_i, t) \mid i \in [n]\} \cup \{(a_i,b_j):i,j\in[n]\}
	\end{align*}
	The capacity $w$ of every edge is $1$. The cost of every edge in $\{(s, a_i ) \mid i \in [n]\} \cup \{(b_i, t) \mid i \in [n]\}$ is $0$. The cost of the edge $(a_i,b_j), \{i,j\}\in[n]^2$ is defined to be the cost of obtaining the preference $\succ_j^{\QQ}$ from $\succ_i^{\PP}$ as given in the cost function $\Pi$. It is clear that the minimum cost of any $s-t$ flow of value $n$ is the minimum cost required by the briber to obtain \QQ as the bribed profile. If that cost is more than $b$, we discard \QQ. Otherwise, we check using \Cref{theorem const is safe} in time $\OO\left((n+m)^m\right)\text{poly}(m,n)$ if \QQ is safe. If we find any safe \QQ which is budget feasible, then we output \yes; otherwise we output \no. Since the number of anonymous profile \QQ is $\OO\left((n+m)^m\right)$, our algorithm runs in time $\OO\left((n+m)^{2m}\right)\text{poly}(m,n)$.
\end{proof}

We now present our results for specific voting rules. We begin with the plurality voting rule.
  
\begin{theorem}\label{plu is safe}
	For plurality, both \DIS and \SBIS are in \Pb.
\end{theorem}

\begin{proof}
	Let $(\CC, \PP, c, \succ_B, \QQ)$  be any instance of \DIS for plurality. For every bad candidate $\bbb \in B$, we define a set $W_\bbb = \{i \mid i \in [n],
	\bbb \text{ is the top candidate in }\suc_i^\PP\}$. Let $|W_\bbb|=n_\bbb$. If $r(\QQ)$ is not $c$, we output \no, as the bribed profile is not successful. Else, we try to find a preference profile, where a bad candidate can win, subject to the constraints of \DIS. To model this, we construct a flow network for every $\bbb\in B$, named $\GG_{\bbb} = (V, E, W)$, to check if it is possible for $\bbb$ to win the election by making some bribed voters not to fully follow the briber's suggestion.
	\begin{align*}
	    V &= \{x_i \mid i \in [n]\} \cup \{y_a \mid a \in \CC\} \cup \{s,t\}\\
	    E &= \{( s,x_i ) \mid i \in [n]\} \cup \{(y_a, t) \mid a \in \CC\} \cup F \cup \{( x_i,y_{\bbb} ) \mid i \in W_{\bbb}\}
	\end{align*}
	The capacities $W$ are as follows. For every $i\in[n]\setminus W_\bbb$, we have an edge $(x_i,y_a)\in F$ if $a$ is the top-ranked candidate in $\succ_i^{\QQ}$ or $\succ_i^{\PP}$. We define the capacity of edge $(y_a,t)$ to be $(n_\bbb-1)$ for every $a\in\CC\setminus\{\bbb\}$ such that $a$ is preferred over $\bbb$ in the tie-breaking rule and $n_\bbb$ for every $a\in\CC\setminus\{\bbb\}$ such that $\bbb$ is preferred over $a$ in the tie-breaking rule. The capacity of the edge $(y_\bbb,t)$ is $n_b$. The capacity of all other edges is $1$. We show that $\bbb$ can be made winner if and only if there is an $s-t$ flow of value $n$ in $\GG_\bbb$. This is because, if there is an $s-t$ flow of value $n$ in $\GG_{\bbb}$, there are exactly $n$ edges between $\{x_i | i \in [n]\}$ and $\{y_a | a \in \CC\}$ having a flow of $1$, which models a legal preference profile. In other words, if edge $(x_i,y_a)$ has flow $1$, it means that in the corresponding preference profile, voter $i$ keeps $a$ at the top of her preference order. By the nature of $\GG_{\bbb}$'s construction, the preference profile would ascertain that $\bbb$ gets enough votes to make her the winner. On the other hand, if $\bbb$ is the winner in some preference profile $\RR$, where $R_i = \succ_i^{\QQ}$, $\ \forall v_i \in \VV_b'$, where $\VV_b' \subseteq \VV_b$ ($\VV_b$ is the set of bribed voters) and $R_i = \succ_i^{\PP}, \forall v_i \in \VV \setminus \VV_b'$. Now, let $T(R_i)$ denote the top candidate in $R_i$. In the corresponding flow instance, there is a flow of $1$ possible through all the edges $(x_i, T(R(i)) \forall i \in[n]$. This is true because only the capacities of the edges incident on $t$ could possibly be violated. However, if this were true for any node $y_a$, then the corresponding candidate $a$ would defeat $x$ (either by having a score of $(n_\bbb+1)$ or by having a score of $n_\bbb$ but by defeating $\bbb$ in tie-breaking). This is a contradiction. Hence we have that there exists an $s-t$ flow of value $n$ in $\GG_{\bbb}$ if and only if $b$ can be made the winner. As the flow cannot exceed $n$, this is in fact the maximum $s-t$ flow.
	
	If for every bad candidate $\bbb\in B$, the value of maximum $s-t$ flow is less than $n$, then we output \yes. Otherwise, we output \no (as in this case, no bad candidate can ever win). Since maximum $s-t$ flow can be computed in $\OO((m+n)n^2)$ using Edmonds-Karp algorithm \cite{edmonds1972theoretical}, and we run at most $m$ instances of it, our algorithm runs in $\OO(m(m+n)n^2)$.
	
 	Any instance of \SBIS for plurality $(\CC, \PP, c, \succ_B, \sss)$ can be mapped to a \DIS instance by considering the fact that if $c$ appears in the top $\sss_i+1$ positions in $\suc_i^\PP$, then $c$ should be placed at the top position in $\suc_i^\QQ$, otherwise $\suc_i^\QQ$ = $\suc_i^\PP$. This allows us to use the above construction to solve an \SBIS instance of plurality in $\OO(m(m+n)n^2)$ time.
\end{proof}

We next show that \SDB and \SSB are polynomial-time solvable for the plurality voting rule.

\begin{theorem}\label{plu safe}
	For plurality, both \SDB and \SSB are in \Pb.
\end{theorem}

\begin{proof}
	Let $(\CC, \PP, c, \succ_B, \Pi, b)$ be an arbitrary instance of \SDB for plurality. Let the plurality winner according to $\PP$
	be $w$. We may assume without loss of generality that the briber asks all the bribed voters to have $c$ as their most preferred candidate. 
	$G$ and $B$ are as defined in \Cref{plu is safe}. Let $s(\bbb)$ denote the initial score of candidate $\bbb$. $\bbb_{max} = max\{s(\bbb) \mid w \succ_t \bbb\}$, according to the tie breaking rule, $\succ_t$. Let $X
	= \{a \in G \mid s(a) \geq \bbb_{max} , a \succ_t \bbb$ or
	$s(a) > \bbb_{max}\}$. Initially $X$ is non-empty since $w \in X$. Let $\lambda$ be the score of $c$ in the constructed bribed profile. For all candidates $a \in \CC$, let us define $\beta_a$ to be $\lambda$ if $c \suc_t a$ and $\beta_a$ to be $\lambda-1$ if $a \suc_t c$. For
	every $x \in X$ and every final score $\lambda$ of $c$ which is in the range $s(c)$ to $n$, we construct a flow
	network $\GG_{x, \lambda}(V, E, C, D, W)$, defining $\gamma_x$ to be $\bbb_{max}$ if $x \succ_t \bbb$, else $(\bbb_{max} + 1)$. Here $D$ is the set of edge demands, $w$ is the set of edge capacities, and $C$ is the set of edge costs. The construction is as follows:
	\begin{align*}
	    V &= \{s, t\} \cup \{ u_i \mid i \in [n] \} \cup \{ y_a \mid a \in \CC \}\\
	    E &= \{(s, u_i) \mid i \in [n]\} \cup \{(y_a , t) \mid a \in \CC\} \cup \{(u_i , y_a) \mid i \in [n], a = c \text{ or} \succ_i^\PP = a \succ \ldots \}
	\end{align*}
	Let the capacities of the edges be as follows. For every candidate $a \in \CC$, the edges $(y_a, t)$ have a capacity of $\beta_a$, and the rest of the edges have capacity $1$ each. Let the costs of the edges be defined as follows. The edges $(u_i,y_c)$ have cost $\pi_i$, only if $c$ does not appear at the top of $\succ_i^{\PP}$.
	The demands (lower bound of flow) on the edges are as follows. The edge $(y_c, t)$ has a demand of $\lambda$, the edge $(y_x, t)$ has a demand of $\gamma_x$ and the rest of the edges have a demand of $0$ each. The basic intuition of the construction revolves around the fact that only the good candidates (in $X$), who initially beat the strongest bad candidate, can ensure safety of the bribed preference profile. This is ensured by edge demands. Also, to ensure that the bribed preference profile is successful, the flow network prevents all other candidates from defeating $c$, with the help of appropriately set edge capacities. The edge costs are used to ensure that the briber's budget is not exceeded
	If there is an $s-t$ flow of value $n$ and cost at most $b$ for some $\GG_{x, \lambda}$, then we output \YES, and the corresponding edge flow values help us construct a safe and successful bribed profile. Otherwise we output \NO. \SSB can be reduced to this problem by defining
	the price of each voter $\pi_i$ to be the cost required to shift $c$ to the top in their ordering. This algorithm is polynomial-time solvable, by running $\OO(mn)$ iterations of the out-of-kilter algorithm \cite{fulkerson1961out}.
\end{proof}

\Cref{plu is safe} and \Cref{plu safe} directly lead to the following result for the shift bribery problems for the $k$-Approval voting rule.

\begin{corollary}\label{k-approval shift}
	For $k$-approval, both \SBIS and \SSB are in \Pb, for every integer $k \in [2, m-1]$.
\end{corollary}

	A bribery instance of $k$-approval is equivalent to a corresponding bribery instance of plurality, if the cost of moving the briber's preferred candidate to the top in plurality is equal to the cost of moving the briber's preferred candidate to any of the top $k$ positions in $k$-approval (for every voter).
\begin{theorem}\label{Theorem veto shift}
For veto, both \DIS and \SBIS are in \Pb.
\end{theorem}

\begin{proof}
	Let $(\CC, \PP, c, \succ_B, \QQ)$ be an arbitrary instance of \DIS for veto.
	Let $\VV_b$, $G$, and $B$ be as defined in \Cref{plu is safe}. For every $\bbb \in B$, define a set $W_\bbb = \{i \mid i \in [n], \bbb$ is the last candidate in both $\succ_i^\PP and \succ_i^\QQ$ in \DB\}. Let $n_\bbb = |W_\bbb|$. If $r(\QQ)$ is not $c$, we output \NO. Else, we construct a
	flow network for every $\bbb$ in $B$, named $\GG_\bbb = (V,E,D,W)$ as follows: 
	\begin{align*}
    V &= \{x_i \mid i \in [n]\} \cup \{y_a \mid a \in \AA\} \cup \{s,t\}\\
	E &= \{(s, x_i) \mid i \in [n]\} \cup \{(y_a, t) \mid a \in \AA\} \cup \{(x_i, y_a) \mid a \in A \setminus \{b\}, \text{ if } y_a \text{ is the last ranked}\\& \text{candidate in} \succ_i^\PP \text{or} \succ_i^\QQ \text{for } i \in [n]\setminus W_b\} \cup \{(x_i ,y_b) \mid i \in W_b\}
	\end{align*}
	Let the demands $D$ (lower bound of flow) on the edges be as follows. For all candidates $a \in \CC \setminus \{\bbb\}$, such that $\bbb$ loses to $a$ in tie-breaking, the edge demand is $n_{\bbb}+1$. Similarly, for all candidates $a \in \CC \setminus \{\bbb\}$, such that $a$ loses to $\bbb$ in tie-breaking, the edge demand is $n_{\bbb}$. Finally the edge $(y_b,t)$ has a demand of $n_\bbb$. Let the capacities $W$ of the edges be as follows. For all candidates $a \in \CC$, the capacities of the edges $(y_a, t)$ is $\infty$. The capacities of the rest of the edges is $1$ each. 
	
	For every bad candidate $\bbb$, $W_\bbb$ represents the absolute value of the best possible score of $\bbb$. Assuming that the bribery is successful, we try to find a partially bribed preference profile in which \bbb is the winner. This is achieved by setting the appropriate edge demands, as mentioned in the construction. If any such partially bribed profile exists, then clearly, we get an $s-t$ flow of $n$. Conversely, if we have a max $s-t$ flow of $n$, we can construct a partially bribed profile in which some $\bbb$ is the winner. Hence, for some $\bbb \in B$, if the max flow in $\GG_\bbb$ is $n$, we return \NO. Else we return \YES. The above algorithm is polynomial-time solvable \cite{cormen2009introduction}.
	
	Any \SBIS instance of veto can be mapped to a \DIS instance by considering the fact that if $c$ appears in the bottom most position in $\suc_i^{\PP}$ and $\sss_i \neq 0$, then $c$ should be raised to the $(m-1)\textsuperscript{th}$ position to get $\suc_i^\QQ$, otherwise $\suc_i^\QQ$ = $\suc_i^\PP$. This allows us to use the above construction to solve an \SBIS instance of plurality in $\OO(m(m+n)n^2)$ time.
\end{proof}

\begin{theorem}\label{theorem veto dol}
	For veto, \SDB is in \Pb.
\end{theorem}             

\begin{proof}
	We use the following observations to construct a polynomial-time algorithm to solve \SDB for veto: 
	
	\begin{enumerate}
		\item The briber will never bribe a voter already vetoing some bad candidate $b_1$, to veto some other bad candidate $b_2$. Had there existed an optimal safe and successful bribery instance such that the briber had bribed a voter to veto $b_2$ instead of $b_1$, if we remove that veto, the new bribery is lower or equal in cost while being safe and successful.
		
		\item Also, it can be shown that the briber will never bribe a voter vetoing some good candidate $g$ (except her top preferred candidate $c$), to veto a bad candidate $b_2$.
		
		\item The number of vetoes received by any bad candidate before bribery is greater than the number of vetoes received by the briber's preferred candidate in the bribed profile. The proof follows from the previous two observations.
	\end{enumerate}
	
	Let $c$ be the briber's most preferred candidate. Since there are only a polynomial number of possible scores for $c$ after bribery (by Observation 3), for each possible score, we construct the least cost bribery by greedily choosing the cheapest voters such that the number of vetoes received by every good candidate is at least as many as the number of vetoes received by $c$, and the number of vetoes received by every bad candidate is greater than the number of vetoes received by $c$. We then return the cheapest possible bribery, if it exists. The algorithm runs in $\OO\left(n^2m\log(n)\log(m)\right)$ time.
\end{proof}

\begin{corollary}\label{theorem veto safe shift}
	For veto, \SSB is in \Pb.
\end{corollary}

The veto rule is equivalent to $k$-approval voting for $k=m-1$. Now, the result follows from \Cref{k-approval shift}.

Next, we take a look at the greedy algorithm for solving \SBIS for simplified Bucklin. Although for simplified Bucklin both \SBIS and \SSB are polynomial-time solvable, their \DB counterparts are not (as we will see in the next section). This is the only rule that we have studied, which is not a scoring rule, but exhibits such behaviour.

\begin{theorem}\label{theorem bucklin shift}
	For simplified Bucklin, \SBIS is in \Pb.
\end{theorem}
\begin{proof}
	We describe a greedy algorithm to solve \SBIS for simplified Bucklin. Consider $(\CC, \PP, c, \succ_B, \sss)$ to be an instance of \SBIS for simplified Bucklin, with $|\CC| = m$. Let the winning candidate according to \PP be $w \in \CC$ and $w \neq c$. 
	
	Let \VV be the set of voters. Let $\ell$ be the simplified Bucklin winning round according to \PP. It is clear that the winning round for any candidate $x \in \CC \setminus \{c\}$ according to \QQ is no smaller than $\ell$. Particularly, the winning round for $w$ is either $\ell$ or $\ell+1$ according to \QQ (as explained in \cite{schlotter2011campaign}). Let $sc_\ell(\PP, a)$ denote the number of votes received by $a \in \CC$, till the $\ell^{\text{th}}$ round according to \PP. 
	\begin{align*}
	\text{Let} \succ_t &= c \succ w \succ \overrightarrow{\CC_1} \succ \overrightarrow{\CC_2}\\
	\text{Let} \succ_B &= c \succ \overrightarrow{\CC_1} \succ w \succ \overrightarrow{\CC_2}, \text{where}\\
	\CC &= \CC_1 \cup \CC_2 \cup \{c,w\}
	\end{align*}
	Note that the above mentioned tie-breaking rule is assumed to simplify the proof. Clearly $\CC_2$ is the set of \textit{bad} candidates. Let $p_i(\XX, a)$ denote the position of candidate $a$ in the $i^{\text{th}}$ vote of some profile $\XX$. If $r(\QQ) \neq c$, we return \NO, as the bribery is not successful. Otherwise, we use the following algorithm for checking safety of \SBIS in simplified Bucklin.

	    Let $B'$ be the complete subset of bad candidates each having number of votes greater than $\lfloor \nicefrac{n}{2}\rfloor$ at the $\ell^{\text{th}}$ level. These are the only candidates who can cause the bribery to be unsafe. Let $b' \in B'$ be a bad candidate who beats all candidates in $B' \setminus \{b'\}$ in tie-breaking. This is the most ``powerful'' bad candidate, who would be the first bad candidate to win.
		Consider a preference profile \RR and initialise it to \PP. Now for every voter $v_i \in \VV$, we do the following iteratively:
		\begin{itemize}
			\item \IF $\exists a \in \CC \setminus \CC_2$ such that $p_i(\RR, a) = \ell$ and $sc_\ell(\RR, a) > \lfloor{\nicefrac{n}{2}}\rfloor$.
			\begin{itemize}
				\item \IF $0 < p_i(\RR,c) - p_i(\RR,a) \leq s_i$, we create a new vote $\succ_i^{\RR_b}$, by shifting $c$ up in $\succ_i^{\RR}$, such that $p_i(\RR_b,a)-p_i(\RR_b,c) = 1$. We then assign $\succ_i^{\RR} = \succ_i^{\RR_b}$.
				\begin{itemize}
				    \item \IF $c$ wins for $\RR$, then \SBIS is a \YES instance; terminate.
				    \item \ELSE, continue.
				\end{itemize}
				\item \ELSE, we keep $\RR$ unchanged.
			\end{itemize}
			\item \ELSE, \SBIS is a \NO instance; terminate.
		\end{itemize}
		If the algorithm terminates at the last \ELSE, it returns a \NO instance, because in this case \RR represents the partially bribed profile which lets some bad candidate (here $b'$ win). This is true because no good candidate is able to cross the majority at the $\el\textsuperscript{th}$ round according to \RR. So, the bribery is unsafe. 
		Otherwise, if the algorithm terminates at the third \IF, when any of these shifts is reverted, we always have some candidate in $\CC_1 \cup \{c, w\}$ who beats all the bad candidates including $b'$ at the $\ell^{\text{th}}$ round, thus ensuring that \SBIS is a \yes instance. The time complexity of this algorithm is $\OO(n)$, because the algorithm does $\OO(n)$ updates (at most one for each voter).
\end{proof}

\begin{theorem}\label{theorem bucklin safe shift}
For the simplified Bucklin voting rule, \SSB is in \Pb, assuming a monotonous price function.
\end{theorem}

\begin{proof}
Let $(\CC, \PP, c, \succ_B, \Pi, b)$ be an instance of \textsc{Safe Shift Bribery} for simplified Bucklin. Let $\pi_i$ be a monotonous pricing function $\forall \pi_i \in \Pi$. Let the winner be $w$ according to $\PP$. Let $\ell$ be the Bucklin winning round in $\PP$. Let $sc_\ell(a)$ denote the number of votes received by $a \in \CC$, till the $\ell^{th}$ round in \PP. Let the tie-breaking rule be $\succ_t$ be 
$$c \succ \overrightarrow{G} \succ \overrightarrow{B}$$ 
where $G$ is the set of good candidates and $B$ is the set of bad candidates. This tie-breaking rule has been assumed for ease of proof, and we can extend the ideas used in this proof for any tie-breaking rule. It can be shown that in any bribed profile, the Bucklin winning round will be either $\ell$ or $\ell+1$ \cite{schlotter2011campaign}. We try to find a safe and successful bribed profile \QQ via \textit{Case 1} followed by \textit{Case 2} (if \textit{Case 1} does not yield any solution).

\textbf{\textit{Case 1:}} Consider the Bucklin winning round in \QQ to be $\ell$. If there is no bad candidate reaching majority (at least $\lfloor n/2 \rfloor +1$ votes) in $\PP$ up to the $\ell^{th}$ round, then we only need to find a successful bribery which is polynomial time solvable according to \cite{schlotter2011campaign}. This is because $w$ reaches majority at the $\ell^{th}$ round, which trivially implies that $w$ also reaches majority in the $(\ell+1)^{th}$ level and defeats all bad candidates, thus ensuring safety. Else, we provide the following algorithm. For each good candidate $g \in G$ which reaches at least the majority in the $\ell^{th}$ round, we create a flow network $\GG_g = (V,E,C,W,D)$. Let us divide the voter set $\VV$ into two sets, $\VV_1$ and $\VV_2$, where $\VV_1 = \{v_i \ | \ v_i \in \VV$, $c$ is in the top $\ell$ positions of $\succ^{\PP}_i\}$ and $\VV_2 = \{v_i \ | \ v_i \in \VV$, $c$ is \textbf{not} in the top $\ell$ positions of $\succ^{\PP}_i\}$. Clearly, there is no incentive for the briber to bribe the voters in $\VV_1$. The construction of $\GG_g$ is as follows: 
\begin{align*}
    V &= \{s,t\} \cup \{u_i \ | \ i \in \VV_2\} \cup \{y_a \ | \ a \in \CC\}\\
    E &= \{(s,u_i) \ | \ i \in [n]\} \cup \{(y_a,t) \ | \ a \in \CC\} \cup \{(u_i,y_c) \ | \ i \in \VV_2, \text{ and the cost of}\\ & \text{shifting } c \text{ to the } \ell^{th} \text{ position in } \succ^{\PP}_i \text{ is } \leq b\} \cup \{(u_i,y_a) \ | \ a \in \CC \setminus \{c\}, i \in \VV_2,\\ &\text{and } a \text{ is at the } \ell^{\text{th}} \text{ position in } \succ^{\PP}_i\}\\
    W &= \{n, \text{for every} (y_a,t) \ | \ a \in \CC\} \cup \{1, \text{ for every other edge}\}\\
    D &= \{\lfloor n/2 \rfloor + 1 - sc_\ell(c) \ | \ (y_c,t)\} \cup \{\lfloor n/2 \rfloor + 1 - sc_{\ell-1}(g) \ | \ (y_g,t)\} \cup \{0 \text{ for}\\ &\text{every other edge}\}\\
    C &= \{cost_i, \text{ for } (u_i,y_c) \ | \ i \in \VV_2, \text{ and } cost_i \text{ is the cost of shifting } c \text{ to the } \ell^{\text{th}}\\ &\text{position in } \succ^{\PP}_i\} \cup \{0, \text{ for all other edges}\}
\end{align*}
    The demand on the edge $(y_c,t)$ ensures that $c$ is the winner in $\QQ$ and the bribery is successful. In a graph $\GG_g$ the demand on the edge $(y_g,t)$ ensures that $g$ always reaches the majority by the $\ell^{th}$ round, even in a partially bribed profile. Since we know that no bad candidate can reach the majority in the $\ell^{th}$ round, the presence of an $s-t$ flow of $n$ using the out-of-kilter algorithm \cite{fulkerson1961out} implies that the bribed preference profile corresponding to the edges with flow $1$, represents a successful and safe bribery. It is easy to see that the converse is also true.
    Hence, if there is an $s-t$ flow of value $n$ and cost at most $b$ for any $\GG_g$ (hence requiring at most $\OO(m)$ iterations), we output \YES, and the corresponding edge flow values help us construct a safe and successful bribed profile. Otherwise, we proceed to \textit{Case 2}.

\textbf{\textit{Case 2:}} 
    Consider the Bucklin winning round to be $\ell+1$ in \QQ. If there exists a bad candidate who reaches at least the majority at the $\ell^{th}$ position, the election is trivially unsafe at the $(\ell+1)^{th}$ level. Consider a bad candidate \bbb who reaches majority at the $\ell^{th}$ level. Assume that each of the bribed voters who have \bbb at the $\ell^{th}$ position of their preference ordering in \PP, stick to their original preference ordering, while those bribed voters who have some good candidate at the $\ell^{th}$ position, comply with the briber's suggested votes. This would imply that no good candidate would reach majority at the $\ell^{th}$ round, making \QQ trivially unsafe. (We know that this can be done because the Bucklin winning round in \QQ is $l+1$). 
    
    If the above is not true, then the election is always safe because $w$ reaches at least the majority in $\PP$ at the $\ell^{th}$ position so it also ensures safety at the $(\ell+1)^{th}$ level. Hence we only need to find a successful bribery within the budget for the $(\ell+1)^{th}$ level, which takes $\OO(mn^3)$ time via the dynamic programming algorithm in Theorem $3$ of \cite{schlotter2011campaign}.
If both the cases fail, then it is not possible obtain a safe and successful bribed profile within the given budget, and we output \NO.
\end{proof}

\subsection{Hardness Results}\label{sec:hard}

We now present our hardness results. We use the \XTC problem which is known to be \NPC~\cite{gareyjohnson}, in many of our hardness proofs.
	
	\begin{definition}[\XTC]
		Given a universe $\UU = \{u_i \mid i \in 3t\}$ of $\ 3t$ elements and a collection $\SS = \{S_i \mid i \in [m]\}$ of subsets of $\ \UU$, where $|S_i | = 3$ for each $i \in [m]$, compute if there exists a set $I \subseteq[m]$ such that $\ \forall \ i, j \in I$ and  $i \neq j$, $S_i \cap S_j = \emptyset$ and $\cup_{i \in I} S_i = \UU$.
	\end{definition}

We show that \DIS is \coNPC for the $k$-approval voting rule for every constant $k\ge 3$. For the following hardness results, we use variants of proof techniques which have been used before~\cite{baumeister2011computational,DeyMN16,DBLP:journals/tcs/DeyMNS21}.

\begin{theorem}\label{Theorem k-approval dol}
	For $k$-approval, \DIS is \coNP-complete, for $k \geq 3$.
\end{theorem}
\begin{proof}
    To see that \DIS belongs to \coNP, any \no instance $(\CC,\PP,c,\suc_B,\QQ)$ can be verified either from the fact that the $k$-approval winner in \QQ is not $c$ or from the existence of a profile $\RR=(R_i)_{i\in[n]}$ such that (i) $R_i=\succ_i^{\PP}$ or $R_i=\succ_i^{\QQ}$ for every $i\in[n]$ and the $k$-approval winner in \RR is less preferred in $\suc_B$ than the $k$-approval winner in \PP. To prove \coNP-hardness, we exhibit a reduction from \XTC to \DIS such that the \XTC instance is a \yes instance if and only if the \DIS instance is a \no instance.

    Let $(\UU= \{u_1, \dots, u_{3t}\},\SS= \{S_1, \dots, S_m\})$ be any instance of \XTC. Without loss of generality, we can assume that $m > t$ by duplicating the sets in \SS. We construct an instance $(\CC,\PP,c,\succ_B,\QQ)$ of \DIS for $k$-approval as follows: 
    \begin{align*}
        \CC &= \UU \cup \{x,w\} \cup \DD \text{, where}\\
        \DD &= \uplus_{i\in[m]} \DD_i^1 \uplus_{i\in[m]} \DD_i^2 \text{ where }|\DD_i^1|=k-3, |\DD_i^2|=k-1\\
        \suc_t &= w \succ c \succ \overrightarrow{\UU \setminus \{c\}} \succ x \succ \overrightarrow{\DD}\\
        \suc_B &= c \succ \overrightarrow{\UU \setminus \{c\}} \succ \overrightarrow{\DD} \succ w \succ x
    \end{align*}
    Here \DD is a set of dummy candidates, who cannot win, $w$ is the winner in the unbribed preference profile, and $c$ is the winner in the bribed preference profile. Using Lemma 4.2 from \cite{baumeister2011computational}, we construct a set of unbribed voters and their votes such that we have $s_{\VV_u}(u_j) = s_{\VV_u}(x) - 2$ for every $j\in[3t]$ and $s_{\VV_u}(w) = s_{\VV_u}(x) - (m-t+1)$. We now describe the set $\VV_b$ of bribed voters, and their bribed and original votes.
	For each $i \in [m]$, we have a bribed voter $v_i \in \VV_b$, with original vote
	$$\succ_i^{\PP} = w \succ \overrightarrow{\DD_i^2} \succ \overrightarrow{\CC \setminus (\{w\} \cup \DD_i^2)}$$
	and bribed vote
	$$\succ_i^{\QQ} = \overrightarrow{S_i} \succ \overrightarrow{\DD_i^1} \succ \overrightarrow{\CC \setminus (S_i \cup \DD_i^1)}$$
	This finishes the description of the reduced \DIS instance. We claim that the \DIS instance for $k$-approval is a \NO instance if and only if the corresponding instance of \XTC is a \YES instance.
	
	$\Longrightarrow$: Suppose the \DIS instance is a \NO instance. We observe that we have only one bad candidate, namely $x$. Hence, there exists a subset $Y\subseteq \VV_b$ such that, if \RR is the profile where voters in $Y$ vote according to \QQ and every other voter votes as the unbribed instance, then the $k$-approval winner of \RR is $x$. We observe that the $k$-approval score of $x$ in \RR is $s_{\VV_u}(x)$. We claim that $|Y|=t.$ We have $|Y| \leq t$, otherwise there exists some $u_j \in \CC$ whose score in \RR is at least $s_{\VV_u}(u_j)$. However, this contradicts our assumption that $x$ is the $k$-approval winner in \RR. Also, $|Y| \ge t$, otherwise the score of $w$ in \RR is at least $s_{\VV_u}(w)+m-t+1$. However, this contradicts our assumption that $x$ is the $k$-approval winner in \RR. Hence, we have $|Y| = t$. Moreover, for $x$ to win, the collection $\{S_i: i\in[m], \succ_i^{\QQ}\in Y\}$ of sets forms an exact cover of $\UU$ as this is the only case which makes the $k$-approval score of $w$ and $u_j$ for every $j\in [3t]$ less than the $k$-approval score of $x$ in \RR. Therefore, \XTC is a \YES instance.

    $\Longleftarrow$: Suppose the \XTC instance is a \yes instance. Let $X \subseteq \SS$ be an exact cover of $\UU$. Let us consider the preference profile \RR where only bribed-voters in $\{v_i \in \VV_b : S_i\in X\}$ vote according to \QQ and others vote according to \PP. The $k$-approval score of the bad candidate $x$ is $s_{\VV_u}(x)$, every candidate in $\{u_j:j\in[3t]\}\cup\{w\}$ is $s_{\VV_u}(x)-1$, and every candidate in \DD is at most $1$ in \RR. Hence, $x$ is the $k$-approval winner in \RR. Thus, the \DIS instance is a \no instance.
\end{proof} 


\begin{corollary} \label{theorem k-veto}
For $k$-veto, \SBIS and \SSB are in \Pb, for $k > 1$; \DIS is \coNPC, for $k \geq 3$.
\end{corollary}

	$k$-veto is in fact $(m-k)$-approval. So, the proof follows from \Cref{Theorem k-approval dol} and \Cref{k-approval shift}.

To get the hardness result for \DIS for Borda, we use the \XTFC problem, which is known to be \NPC \cite{faliszewski2008copeland}.
\begin{definition}[\XTFC]
	Given a universe $\UU = \{u_1, u_2, \ldots, u_{4m/3}\}$ of $4m/3$ elements and a collection $\SS = \{S_i \mid i \in [m]\}$ of subsets of $\ \UU$, where $|S_i| = 4$ for each $i \in [m]$ and where each $u_i \in \UU$ is in exactly $3$ sets $S_j, S_k, S_l$ for some $j, k, l \in [m]^3$, compute if there exists a set $I \subseteq [m]$ such that $\ \forall i, j \in I$, $i \neq j$, $S_i \cap S_j = \emptyset$ and
	$\cup_{i \in I} S_i = \UU$.
\end{definition}

\begin{theorem}\label{theorem borda dol}
	For Borda, \DIS is \coNPC.
\end{theorem}

\begin{proof}
	Firstly \DIS for Borda belongs to \coNP, provable in a similar fashion as in \Cref{Theorem k-approval dol}. To prove \coNP-hardness, we demonstrate a reduction from \XTFC to \DIS such that the \XTFC instance is a \yes instance if and only if the \DIS instance is a \no instance. 
	
	Let $(\UU=\{u_1, u_2, \ldots, u_{\nicefrac{4m}{3}}\},\SS=\{S_i \mid i \in [m]\})$ be an instance of \XTFC. We construct an instance $(\CC,\PP, c, \succ_B, \QQ)$ of \DIS for Borda as follows:
    \begin{align*}
     \CC &= \UU \cup \{w,x\} \cup \DD \text{, where } \{w,x\} \notin \UU\\
     \succ_t &= w \succ c \succ \overrightarrow{\UU \setminus \{c\}} \succ x \succ \overrightarrow{\DD}\\
     \succ_B &= c \succ \overrightarrow{\UU \setminus \{c\}}  \succ \overrightarrow{\DD} \succ w \succ x
    \end{align*}
	 Here $c = u_j$ for some $j \in [\nicefrac{4m}{3}]$ is the winner in the bribed preference profile, and $w$ is the winner according in the unbribed preference profile. \DD is a set of $\OO(m^2)$ dummy candidates (who can never win). These candidates are added because the proof requires that $|\CC|=\beta m^2+\gamma$, where $\beta \geq \nicefrac{5}{3}$ and $\gamma \geq 15$.
	Consider $\VV$ as the set of voters. Let $\VV_b$ be the set of $m$ bribed voters (each corresponding to an $S_i \in \SS$), and $\VV_u$ be the rest of the voters. Using Lemma 4.2 from \cite{baumeister2011computational}, we can construct the election such that if we consider only the votes in $\VV_u$, the scores of the candidates are in the order $s_{\VV_u}(x) > s_{\VV_u}(\UU) > s_{\VV_u}(w) > s_{\VV_u}(\DD)$, with
	\begin{align*}
	s_{\VV_u}(x) - s_{\VV_u}(w) &= (|\CC|-1) \left(\nicefrac{2m}{3} + 1\right) + \left(\nicefrac{m}{3} - 1\right)\\
	s_{\VV_u}(x) - s_{\VV_u}(u_i) &= 2\beta m^2 + 2\gamma - 12
	\end{align*}
	For each $i \in [m]$, we have a bribed voter $v_i \in \VV_b$, with original vote
	$$\succ_i^{\PP}=w \succ \overrightarrow{\DD} \succ \overrightarrow{\UU \setminus S_i} \succ \overrightarrow{S_i} \succ x$$
	and bribed vote
	$$\succ_i^{\QQ}=\overrightarrow{S_i} \succ \overrightarrow{\DD} \succ \overrightarrow{\UU \setminus S_i} \succ w \succ x$$
	
	It can be shown that an instance of \DIS for Borda is a \NO instance if and only if the corresponding instance of \XTFC is a \YES instance.
    
	$\Longrightarrow$: Assume the instance of \DIS for Borda is a \NO instance. This means that $x$ is the winner in some preference profile $\RR$ where a subset of the bribed voters vote according to briber's suggested preference orders. We now show that this preference profile corresponds to a solution of the \XTFC problem. Let $Y \subseteq \VV_b$, be the set of voters who vote according to their corresponding preference orders in $\QQ$, so for all $v_i \in Y, \succ_i^\RR = \succ_i^\QQ$ and for all $v_i \notin Y, \succ_i^\RR = \succ_i^\PP$. We next show that for $x$ to win, $|Y| = t$. Clearly, $|Y| \leq t$, else for some $u_i \in \UU$, the score of $u_i$ and $x$ will atleast become equal and $x$ would lose. Also, $|Y| \geq t$, else the score of $w$ will atleast become equal to the score of $x$, and $x$ will lose. Hence, $|Y| = t$. Moreover, for $x$ to win, the preference orders of the voters in $Y$ must correspond to an exact cover of $\UU$. This is the only case where the score of $w$ as well as those of $u_i \ \forall u_i \in \UU$ are exactly $1$ less than that of $x$. Therefore, \XTFC is a \YES instance.
	
	$\Longleftarrow$: Assume the \XTFC instance is a \YES instance. To show that \DIS for Borda is a \NO instance, it is enough to show that a bad candidate ($x$), wins according to some partially bribed profile $\RR$. Let $X \subseteq \SS$ be the set of quadruplets that form an exact cover of $\UU$. Let us consider the preference profile \RR where only bribed voters in $\{v_i \in \VV_b : S_i\in X\}$ vote according to \QQ and others vote according to \PP. Clearly $|X| = \nfrac{m}{3}$. In this case, the final score of $u_i$, $\forall u_i \in \UU$, is less than the score of $x$. Moreover, the final score of $w$ is also less that the score of $x$. Hence, there exists a preference profile $R$, where a bad candidate ($x$) wins, implying that \DIS for Borda is a \NO instance.
	\end{proof}

\begin{theorem}\label{theorem borda shift}
	For Borda, \SBIS is \coNPC
\end{theorem}

\begin{proof}
Firstly, \SBIS for Borda belongs to \coNP, provable using a method similar to \Cref{Theorem k-approval dol} by showing a certificate (possibly involving partial shifts) which demonstrates that the bribery is not successful, or that it is successful but not safe. To prove \coNP-hardness, we demonstrate a reduction from \XTC to \SBIS such that the \XTC instance is a \YES instance if and only if the \DIS instance is a \NO instance.
	Let $(\UU = \{u_1, \dots, u_{3t}\}, \SS = \{S_1, \dots, S_m\})$ be an instance of \XTC. Without loss of generality, we can assume that $m \geq t$. We construct an instance $(\CC, \PP, c, \succ_B, \QQ)$ of \SBIS for Borda as follows: 
	\begin{align*}
	\CC &= \UU \cup \DD_1 \cup \{x,c,w,d\}, \text{ where } \{x, c, w, d\} \notin \UU\\
	\DD &= \DD_1 \cup \{d\} \text{ is a set of dummy candidates, where } |\DD_1| = t-1\\
	\succ_t &= c \succ \UU \succ w \succ \DD \succ x\\
	\succ_B &= c \succ \UU \succ \DD \succ w \succ x
	\end{align*}
	Here, $w$ is the winner in the unbribed preference profile and $c$ is the winner in the bribed preference profile. We now describe the voters. \VV consists of $2m + 2t-2$ voters of the following types: There are $m$ voters (corresponding to each $S_i \in \SS$) each of type (i)-1 and (i)-2;
	
	(i)-1: $\ w \succ \overrightarrow{S_i} \succ c \succ \overrightarrow{\UU \setminus S_i} \succ x \succ \overrightarrow{\DD_1} \succ d$
	
	(i)-2: $\ x \succ \overleftarrow{\UU \setminus S_i} \succ c \succ \overleftarrow{S_i} \succ w \succ \overleftarrow{\DD_1} \succ d$
	
	\noindent There is a single voter, each of type (ii)-1, (ii)-2; 
	
	(ii)-1: $\ w \succ d \succ x \succ \overrightarrow{\UU} \succ \overrightarrow{\DD_1} \succ c$
	
	(ii)-2: $\ w \succ \overleftarrow{\UU} \succ x \succ c \succ d \succ \overleftarrow{\DD_1}$
	
	\noindent There are $t-2$ voters, each of type (iii)-1 and (iii)-2;
	
	(iii)-1: $\ \overrightarrow{\UU} \succ d \succ c \succ w \succ x \succ \overrightarrow{\DD_1}$
	
	(iii)-2: $\  x \succ d \succ c \succ w \succ \overleftarrow{\UU} \succ \overleftarrow{\DD_1}$.
	
	It is easy to observe that each non-dummy candidate gets a constant score (say $L$) from the voters of types (i)-1 and (i)-2. Let $s(a)$ denote the Borda score of a candidate $a \in \CC$. Therefore, initially we have:
\begin{align*}
    &s(d_i) \ll s(a), \forall d_i \in \DD, \forall a \in \CC \setminus \DD\\
    &s(w)-s(c) = 5t\\
    &s(x)-s(c) = 3t+1\\
    &s(u_i)-s(c) = 3t+1, \forall u_i \in \UU.
\end{align*}
	Here $\VV_b$ is equal to the set (i)-1 of $m$ voters. For each $i \in [m]$, we have a bribed voter $v_i \in \VV_b$ with original vote
	$$\succ_i^{\PP} = w \succ \overrightarrow{S_i} \succ c \succ \overrightarrow{U \setminus S_i} \succ x \succ \overrightarrow{\DD_1} \succ d$$
	and bribed vote
	$$\succ_i^{\QQ} = c \succ w \succ \overrightarrow{S_i} \succ \overrightarrow{\UU \setminus S_i} \succ x \succ \overrightarrow{\DD_1} \succ d$$
	
	 It is easy to observe that partial shifts cannot be allowed in these votes as then, $x$ will lose to either $c$ or $w$ and the bribery will always be safe. Now, it is easy to show that the given instance of \SBIS for Borda is a \NO instance if and only if the instance of X3C is a \YES instance.
	
		$\Longrightarrow$: Assume the instance of \SBIS for Borda is a \NO instance, implying that a bad candidate ($x$) is the winner in some preference profile \RR where a subset of the bribed voters vote according to the briber's suggested preference orders.
		We now show that this preference profile corresponds to a solution of the \XTC problem. Let $Y \subseteq \VV_b$, be the set of voters who vote according to $\succ_i^{\QQ}$  ($\forall v_i \in Y$), so for all $v_i \in Y, \succ_i^\RR = \succ_i^\QQ$ and for all $v_i \notin Y, \succ_i^\RR = \succ_i^\QQ$.	We show that for $x$ to win, $|Y| = t$. 
		Clearly $|Y| \leq t$, else $s(c)$ will increase by at least $4t+4$, thus $c$ will defeat $x$. Also, $|Y| \geq t$, else $s(w)$ will decrease by at most $t-1$. Thus she will defeat $x$. Hence, $|Y| = t$.
		Now, for $x$ to win, the preference orders of the voters in $Y$ must correspond to an exact cover of $\UU$, otherwise $\exists u_i \in U$, whose score will not decrease. That $u_i$ will defeat $x$ in tie-breaking. This implies that \XTC is a \YES instance. 
		
		$\Longleftarrow$: Assume the instance of \XTC is a \YES instance. To show that \SBIS for Borda is a \NO instance, it is equivalent to show that a bad candidate ($x$), wins in some partially bribed profile. Let $X \subseteq \SS$ be the set of triplets that form an exact cover of $U$. Clearly $|X| = t$.
		Now, if for every $S_i \in X$, we use the preference order $\succ_i^{\QQ}$ instead of $\succ_i^{\PP}$, $s(w)$ decreases by $t$, which means $s(x)-s(w)$ becomes $1$. Also, $s(c)$ increases by $4t$, meaning $s(x)-s(c)$ becomes 1. Finally, $s(u_i)$ decreases by $1, \forall i \in [3t]$, which means that $s(x)-s(u_i)$ becomes $1, \forall i \in [3t]$, implying $x$ defeats $u_i$, $\forall i \in [3t]$. Therefore, $x$ becomes the Borda winner. Hence, there exists a possible preference profile, where a bad candidate ($x$) wins. Hence \SBIS for Borda is a \NO instance.
\end{proof}

Next, we show that for simplified Bucklin for \DIS is \coNPC.

\begin{theorem}\label{theorem bucklin dol}
	For simplified Bucklin, \DIS is \coNPC.
\end{theorem}

\begin{proof}
To see that \DIS for simplified Bucklin belongs to \coNP, any \no instance $(C,\PP,c,\suc_B,\QQ)$ can be verified either from the fact that the simplified Bucklin winner in \QQ is not $c$ or from the existence of a profile $\RR=(R_i)_{i\in[n]}$ such that (i) $R_i=\succ_i^{\PP}$ or $R_i=\succ_i^{\QQ}$ for every $i\in[n]$ and the simplified Bucklin winner in \RR is less preferred in $\suc_B$ than the simplified Bucklin winner in \PP. To prove \coNP-hardness, we demonstrate a reduction from \XTC to \SBIS such that the \XTC instance is a \YES instance if and only if the \DIS instance is a \NO instance.

Let $(\UU = \{u_1, \dots, u_{3t}\}, \SS = \{S_1, \dots, S_m\})$ be an instance of \XTC. Without loss of generality, we can assume that $m \geq t$. We construct an instance $(\CC, \PP, c, \succ_B, \QQ)$ of \DIS for simplified Bucklin as follows: 
\begin{align*}
\CC &= \UU \cup x \cup w \cup \DD, \text{ where }\\
\DD &= \uplus_{i\in[m]} \DD_i^1 \uplus_{i\in[m]} \DD_i^2 \text{, where }|\DD_i^1|=\ell-3,\text{ }  |\DD_i^2|=\ell-1\\
\succ_t &= w \succ \UU \succ x \succ \DD\\
\succ_B &= c \succ \UU \setminus \{c\} \succ \DD \succ w \succ x
\end{align*}
Here $c = u_j$ for some $j \in [3t]$, is the winner in the bribed preference profile, and $w$ is the winner in the unbribed preference profile. Also, \DD is a set of $2m(l-2)$ dummy candidates whose scores are so low that they can never win. We assume that $\ell$ is the \textit{Bucklin winning round} in $\PP$. First we describe the set $\VV_u$ of unbribed voters and their votes. Using Lemma 4.2 from \cite{baumeister2011computational}, we can construct the election such that if we consider only the votes in $\VV_u$, the $\ell$-approval scores of the candidates are in the order $s_{\VV_u}(x) > s_{\VV_u}(\UU) > s_{\VV_u}(w) > s_{\VV_u}(\DD)$. $\VV_u$ is then constructed as follows:
$\VV_u$ consists of $\nicefrac{|\PP|}{2} + 1$ voters whose preference orders. All these votes have $x$ in the $\ell\textsuperscript{th}$ position. The first $\nicefrac{|\PP|}{2} - m + t$ votes have $w$ in the $(\ell-1) \textsuperscript{th}$ position. The first $\nicefrac{|\PP|}{2} - 1$ votes have $U$ in the first $3t$ positions. The remaining positions are filled by dummy candidates. 
We thus ensure that $x$ receives $\nicefrac{|\PP|}{2} + 1$ votes till round $l$, where $\ell \geq |\UU|+2$.
Every $u_i \in U$ receives $\nicefrac{|\PP|}{2} - 1$ votes till round $l$, and
$w$ receives $\nicefrac{|\PP|}{2} - m + t$ votes.
We now describe the set $\VV_b$ of bribed voters, and their bribed and original votes.
	For each $i \in [m]$, we have a bribed voter $v_i \in \VV_b$ with original vote
	$$\succ_i^{\PP} = w \succ \overrightarrow{\DD_i^2} \succ \overrightarrow{\CC \setminus (\{w\} \cup \DD_i^2)}$$
	and bribed vote
	$$\succ_i^{\QQ} = \overrightarrow{S_i} \succ \overrightarrow{\DD_i^1} \succ \overrightarrow{\CC \setminus (S_i \cup \DD_i^1)}$$
Now we show that the given instance of \DIS for simplified Bucklin is a \NO instance if and only if the instance of \XTC is a \YES instance.

$\Longrightarrow$: Assume the instance of \DIS for simplified Bucklin is a \NO instance. This means that $x$ is the winner, in a particular preference profile \RR where a subset of the bribed voters vote according to the briber's suggested preference orders. We now show that this voting profile corresponds to a solution of the X3C problem. Let $Y \subseteq \VV_b$, be the set of voters who vote according to $\succ_i^{\QQ}$  ($\forall v_i \in Y$), so for all $v_i \in Y, \succ_i^\RR = \succ_i^\QQ$ and for all $v_i \notin Y, \succ_i^\RR = \succ_i^\QQ$. We show that for $x$ to win, $|Y| = t$. 
Clearly $|Y| \leq t$, else the number of votes received by some $u_i \in \UU$, will increase by at least $2$ and she will win before the $l \textsuperscript{th}$ round, thus defeating $x$. Also, $|Y| \geq t$, else the number of votes received by $w$ will increase by at least $m-t+1$, and she will win in the $(\ell-1) \textsuperscript{th}$ round, thereby defeating $x$. Hence, $|Y| = t$.
Now, for $x$ to win, the preference orders of the voters in $Y$ must correspond to an exact cover of $\UU$, which would force the simplified Bucklin winning rounds of $w$ as well as every $u_i$ to be greater than $\ell$. Therefore \XTC is a \YES instance.

$\Longleftarrow$: Assume the instance of \XTC is a \YES instance. To show that \DIS for simplified Bucklin is a \NO instance, it is equivalent to show that there exists some preference profile where a bad candidate ($x$), wins. Let $X \subseteq \SS$ be the set of triplets that form an exact cover of $U$. Clearly $|X| = t$.
Now, if for every $S_i \in X$, we use the preference order $\succ_i^{\QQ}$ instead of $\succ_i^{\PP}$, the $\ell$-approval score of every $u_i$, becomes $1$ less than the score of $x$, and the $\ell$-approval score of $w$ also becomes $1$ less that the score of $x$. Hence $x$ is the winner. Thus, there exists a possible voting profile, where a bad candidate ($x$) wins.
Hence \DIS for simplified Bucklin is a \NO instance.
\end{proof}


 	Next, we show the hardness results for some tournament-based rules. For \DIS is safe, we obtain a generalized hardness result which is applicable for any Condorcet-consistent voting rule.

\begin{theorem}\label{theorem condorcet dol}
	For Copeland and maximin, \DIS is \coNPC.
\end{theorem}

\begin{proof}
    We will prove the above theorem for Copeland, and it is observed that the same construction works for the maximin rule.
    Firstly \DIS for Copeland belongs to \coNP, provable in a similar fashion as in \Cref{Theorem k-approval dol}. To prove \coNP-hardness, we demonstrate a reduction from \XTC to \DIS such that the \XTC instance is a \YES instance if and only if the \DIS instance is a \NO instance. 
	Let $(\UU=\{u_1, \dots, u_{3t}\},\SS=\{S_1, \dots, S_m\})$ be an instance of \XTC. Without loss of generality, we can assume that $m > 2t+1$. We construct an instance $(\CC, \PP, c, \succ_B, \QQ)$ of \DIS for Copeland as follows: 
	\begin{align*}
	\CC &= \UU \cup \{x, w, c\}\\
	\succ_t &= c \succ \overrightarrow{\UU} \succ w \succ x\\
	\succ_B &=  c \succ \overrightarrow{\UU} \succ w \succ x
	\end{align*}
	Here $c$ is the winner in the bribed preference profile, and $w$ is the winner in the unbribed preference profile.
	Let $\VV_b$ be the set of $m$ bribed voters and let $\VV_u$ be the rest of the voters.
	
	We construct a set $\VV_u$ of $m + 8t - 15$ voters of the following types:\\
	There are $(m-3)$ voters of type (i); 

	 (i): $\overrightarrow{\UU} \succ c \succ x \succ w$
	
	\noindent There is a single voter, each of type (ii)-1 and (ii)-2; 

    (ii)-1: $x \succ \overrightarrow{\UU} \succ c \succ w$
	
	(ii)-2: $w \succ \overrightarrow{\UU} \succ c \succ x$

	\noindent There are $(t-2)$ voters, each of type (iii)-1, (iii)-2, (iv)-1, (iv)-2, (v)-i and (v)-ii; 
	
	(iii)-1: $x \succ w \succ \overrightarrow{\UU} \succ c$
	
	(iii)-2: $c \succ w \succ \overrightarrow{\UU} \succ x$
	
	(iv)-1: $x \succ w \succ c \succ \overrightarrow{\UU}$
	
	(iv)-2: $\overleftarrow{\UU} \succ w \succ c \succ x$
	
	(v)-1: $\overrightarrow{\UU} \succ w \succ x \succ c$
	
	(v)-2: $c \succ w \succ x \succ \overleftarrow{\UU}$
	
	\noindent There are $(t-1)$ voters of type (vi)-1 and (vi)-2; 
	
	(vi)-1: $\overrightarrow{\UU} \succ x \succ c \succ w$
	
	(vi)-2: $w \succ x \succ c \succ \overleftarrow{\UU}$. 
	
	Thus, from the above preference orders, we get the following pairwise scores (based on the votes of the unbribed voters):
	\begin{align*}
	vs(x,w) &= m - 2t + 1\\
	vs(c,x) &= m - 2t - 1\\
	vs(u_i, x) &= m - 3,  \forall i \in [3t]\\
	vs(u_i, w) &= m - 2t + 1,  \forall i \in [3t]\\ 
	vs(c, w) &= m - 2t + 1\\
	vs(u_i, c) &= m - 1,  \forall i \in [3t].
	\end{align*}
	Clearly, the absolute Copeland scores for the candidates, as calculated from the unbribed portion of the votes, are as follows:
	\begin{align*}
		s(x) &= 1, \text{ as } x \text{ defeats } w.\\
		s(c) &= 2, \text{ as } c \text{ loses to } u_i \ \forall i \in [3t], \text{ and defeats } w \text{ and } x.\\
		s(u_i) &\geq 2, \forall i \in [3t], \text{ as every } u_i \text{ beats at least 2 candidates } (x \text{ and } c).\\
		s(w) &= 0, \text{ as } w \text{ loses to every other candidate}.
    \end{align*}
	For each $i \in [m]$, we have a bribed voter $v_i \in \VV_b$, with original vote
	$$\succ_i^{\PP} = w \succ x \succ c \succ \overrightarrow{\UU}$$
	and bribed vote
	$$ \succ_i^{\QQ} = c \succ \overrightarrow{S_i} \succ x \succ \overrightarrow{\UU \setminus S_i} \succ w $$
	Now, it is easy to show that an instance of \DIS for Copeland is a \NO instance if and only if the corresponding instance of \XTC is a \YES instance.
	
	$\Longrightarrow$: Assume the instance of \DIS for Copeland is a \NO instance. This means that $x$ is the winner, in a particular preference profile \RR, where a subset of the bribed voters vote according to the briber's suggested preference orders. We now show that this preference profile corresponds to a solution of the \XTC problem. Let $Y \subseteq \VV_b$, be the set of voters who vote according to $\succ_i^{\QQ}$  ($\forall v_i \in Y$), so for all $v_i \in Y, \succ_i^\RR = \succ_i^\QQ$ and for all $v_i \notin Y, \succ_i^\RR = \succ_i^\QQ$. We show that for $x$ to win, $|Y| = t$. Clearly $|Y| \leq t$, else $c$ would be the Condorcet (and Copeland) winner. This is because $vs(c,x)$ would become at least $1$, $vs(c,w)$ would become at least $3$ and $\forall i \in [3t], \ vs(c,u_i)$ would become at least $1$. Also, $|Y| \geq t$, else $w$ would be the Condorcet (and Copeland) winner. This is because $\forall i \in [3t], vs(w,u_i)$ would become at least $1$, $vs(w,x)$ would become at least $1$ and $vs(w, c)$ would also become at least $1$. Hence, $|Y| = t$. Now, for $x$ to win, the preference orders of the voters in $Y$ must correspond to an exact cover of $\UU$. Else, there will be at least one $u_k$, for $k \in [3t]$, which is covered at least twice. This would then imply $vs(u_k,x) \geq 1$, $vs(u_i,c) = -1$, $\forall i \in [3t]$, $vs(x,w) = 1$, $vs(c,x) = -1$, $vs(x,w) = 1$ and  $vs(u_i,w) = 1$, $\forall i \in [3t]$. We therefore see that $s(c) \geq s(x)$, $s(c) \geq s(u_i)$, $\forall i \in [3t]$, and $s(c) > s(w)$. Also, in $\succ_t$, $c$ is at the first position. Therefore $c$ becomes the Copeland winner. On the other hand, if the votes do correspond to an exact cover, then $vs(x,u_i) = 1$, $\forall i \in [3t]$, $vs(x,c) = 1$ and $vs(x,w) = 1$. We therefore see that in this case, $x$ becomes the Copeland winner. Hence \XTC is a \YES instance.
		
	$\Longleftarrow$: Assume the \XTC instance is a \YES instance. To show that \DIS for Copeland is a \NO instance, it is equivalent to show that a bad candidate ($x$) is the winner in a some preference profile where a subset of the bribed voters vote according to briber's suggested preference orders. Let $X \subseteq \SS$ be the set of triplets that form an exact cover of $U$. Clearly $|X| = t$. Now, if for every $S_i \in X$, we use the preference order $\succ_i^{\QQ}$ instead of $\succ_i^{\PP}$, $vs(c,x)$ becomes $-1$ and $vs(x,w)$ becomes $1$, implying that $x$ defeats both $c$ and $w$. Also, $vs(u_i,x), \forall i \in [3t]$ becomes $-1$, implying $x$ defeats $u_i$, $\forall i \in [3t]$. Therefore, $x$ becomes the Condorcet (and thus also Copeland) winner. Hence, there exists a possible voting profile, where a bad candidate (i.e. $x$) wins.
		Therefore, \DIS for Copeland is a \NO instance.
\end{proof}

\begin{theorem}\label{Theorem Copeland}
	For Copeland, \SBIS is \coNPC.
\end{theorem}

\begin{proof}
	 Firstly, \SBIS for Copeland belongs to \coNP, provable using a method similar to \Cref{Theorem k-approval dol} by showing a certificate (possibly involving partial shifts) which demonstrates that the bribery is not successful, or that it is successful but not safe. Let $(\UU=\{u_1, \ldots, u_{3t}\},\SS=\{S_1, \ldots, S_m\})$ be an instance of \XTC. Without loss of generality, we can assume that $m \geq t$, where $t$ is an odd natural number. We construct an instance $(\CC, \PP, c, \succ_B, \QQ)$ of \SBIS for Copeland as follows:
	 \begin{align*}
	 \CC &= \UU \cup \DD \cup \{w,x,c,e\}\\
	 \succ_t &= c \succ w \succ e \succ \overrightarrow{\UU} \succ \overrightarrow{\DD} \succ x\\
	 \succ_B &= c \succ e \succ \overrightarrow{\UU} \succ \overrightarrow{\DD} \succ w \succ x
	 \end{align*}
	  Here \DD is the set of $\nicefrac{(3t-1)}{2}$ dummy candidates. Let $A$ be the set $\UU \cup \{w, x\}$, $B$ be the set $\DD$ (consisting of dummy candidates) and $C$ be the set $\{e, c\}$.  Note that $w$ is the winner in the unbribed preference profile and $c$ is the winner in the bribed preference profile. Let $\succ_t$ be $c \succ w \succ e \succ \overrightarrow{\UU} \succ \overrightarrow{\DD} \succ x $, and $\succ_B$ be $c \succ e \succ \overrightarrow{\UU} \succ \overrightarrow{\DD} \succ w \succ x$.

	We now construct the voter set $\VV$ consisting of $2m+8t+7$ voters having the following types of preference orders in \PP: 
	
	\noindent There are $m$ voters ($\forall i \in [m]$), each of types (i)-1 and (i)-2; 
	
	(i)-1: $x \succ w \succ \overrightarrow{S_i} \succ e \succ c \succ \overrightarrow{\UU \setminus S_i} \succ \overrightarrow{\DD}$
	
	(i)-2: $\overleftarrow{\DD} \succ \overleftarrow{U \setminus S_i} \succ c \succ e \succ \overleftarrow{S_i} \succ w \succ x$
	
	\noindent There are $\nicefrac{(t+1)}{2}$ voters each, of types (ii)-1 and (ii)-2; 
	
	(ii)-1: $\overrightarrow{\UU} \succ x \succ w \succ c \succ e \succ \overrightarrow{\DD}$
	
	(ii)-2: $\overleftarrow{\DD} \succ c \succ e \succ \overleftarrow{\UU} \succ x \succ w$
	
	\noindent There are $\nicefrac{(t+3)}{2}$ voters each, of types (iii)-1 and (iii)-2; 
	
	(iii)-1: $\overrightarrow{\UU} \succ x \succ c \succ w \succ e \succ \overrightarrow{\DD}$
	
	(iii)-2: $\overleftarrow{\DD} \succ e \succ c \succ w \succ x \succ \overleftarrow{\UU}$
	
	\noindent There are $\nicefrac{(3t+1)}{2}$ voters each, of type (iv)-1 and (iv)-2;
	
	(iv)-1: $x \succ w \succ \overrightarrow{\UU} \succ c \succ e \succ \overrightarrow{\DD}$
	
	(iv)-2: $\overleftarrow{\DD} \succ e \succ \overleftarrow{\UU} \succ c \succ w \succ x$
	
	\noindent We define $3t+2$ votes of type (v)-j as follows: Consider a fixed ordering $\overrightarrow{A_1}$ on the set $A$. To get $\overrightarrow{A_j}$,  $\forall j \in [3t+2]$, $\overrightarrow{A_1}$ is rotated clockwise by $j-1$ candidates; 
	
	(v)-j: $\overrightarrow{A_j} \succ \overrightarrow{\DD} \succ e \succ c$. 
	
Considering the above preference orders, we have: 
\begin{align*}
vs(e,c) &= 2t+1 \\
vs(w,c) &= 2t-1 \\
vs(u_i,c) &= 1, \forall u_i \in \UU
\end{align*}
Clearly, the absolute Copeland scores for the candidates, as calculated from the preference orders defined above, are as follows:
\begin{align*}
s(e) &= 1 \\
s(c) &= 0 \\
s(a_i) &= 3t + 2, \forall a_i \in A \\
s(d) &\leq \nicefrac{(3t + 1)}{2}, \forall d \in \DD
\end{align*}
Here $\VV_b$ is equal to set (i)-1 of $m$ voters. For each $i \in [m]$, we have a bribed voter $v_i \in \VV_b$ (corresponding to type (i)-1) with original vote
$$\succ_i^{\PP}=x \succ w \succ \overrightarrow{S_i} \succ e \succ c \succ \overrightarrow{\UU \setminus S_i} \succ \overrightarrow{\DD}$$
and bribed vote
$$\succ_i^{\QQ}=x \succ c \succ w \succ \overrightarrow{S_i} \succ e \succ \overrightarrow{\UU \setminus S_i} \succ \overrightarrow{\DD}$$

Note that shift is allowed only up to the $2^{\text{nd}}$ position. It is important to observe that partial shifts cannot be allowed in the bribed votes as then, $x$ will lose to either $c$ or $w$ and the bribery will always be safe. Using this fact, it can be shown that an instance of \SBIS for Copeland is a \NO instance if and only if the corresponding instance of \XTC is a \YES instance.

	$\Longrightarrow$: Assume the instance of \SBIS for Copeland is a \NO instance, implying that a bad candidate ($x$) wins in a particular preference profile \RR, where a subset of the bribed voters vote according to the briber's suggested preference orders. We now show that this preference profile corresponds to a solution of the \XTC problem. Since partial shifts do not help, any bribed voter $v_i \in \VV_b$ can only vote according to $\succ_i^\PP$ or $\succ_i^\QQ$. Let $Y \in \VV_b$ be the set of voters who vote according to $\succ_i^{\QQ}$ ($\forall v_i \in Y$), so for all $v_i \in Y, \succ_i^\RR = \succ_i^\QQ$ and for all $v_i \notin Y, \succ_i^\RR = \succ_i^\QQ$. We now show that for $x$ to win, $|Y| = t$. Clearly, $|Y| \leq t$, else $s(c) = s(x)$, and in tie-breaking, $c$ wins.
	Also, $|Y| \geq t$, else $w$ wins. Now, for $x$ to win, the preference orders of the voters in $Y$ must correspond to an exact cover of \UU, otherwise $\exists u_i \in \UU$, such that $vs(u_i,c)$ does not become negative, $s(u_i)$ does not decrease, and in tie-breaking, that particular $u_i$ wins against $x$. Hence, \XTC is a \YES instance.
	
	$\Longleftarrow$: Assume the instance of \XTC is a \YES instance. To show that \SBIS for Copeland is a \NO instance, it suffices to show that there exists a preference profile where a bad candidate ($x$), wins. Let $X \subseteq \SS$ be the set of triplets that form an exact cover of $U$. Clearly $|X| = t$.
	Now, if for every voter corresponding to the sets $S_i \in X$, we use the preference order $\succ_i^{\QQ}$ instead of $\succ_i^{\PP}$, $vs(w,c)$ decreases by $2t$ and becomes $-1$. Also, $vs(u_i,c)$ decreases by $2$ and becomes $-1, \forall i \in [3t]$. Lastly, $vs(e,c)$ decreases by $2t$ and becomes $1$. Therefore the final Copeland score of $x$ is greater than the scores of $c$, $e$ and $u_i, \forall i \in [3t]$. So, $x$ becomes the Copeland winner. Hence, there exists a possible voting profile, where a bad candidate wins. Thus, \SBIS for Copeland is a \NO instance. 
\end{proof}

\begin{theorem}\label{theorem maximin shift}
	For maximin, \SBIS is \coNPC.
\end{theorem}

\begin{proof}
Firstly, \SBIS for maximin belongs to \coNP, provable using a method similar to \Cref{Theorem k-approval dol} by showing a certificate (possibly involving partial shifts) which demonstrates that the bribery is not successful, or that it is successful but not safe. To prove \coNP-hardness, we demonstrate a reduction from \XTC to \SBIS such that the \XTC instance is a \YES instance if and only if the \DIS instance is a \NO instance.

Let $(\UU= \{u_1, \dots, u_{3t}\},\SS= \{S_1, \dots, S_m\})$ be any instance of \XTC. Without loss of generality, we can assume that $m \geq t$. We construct an instance $(\CC, \PP, c, \succ_B, \QQ)$ the \SBIS problem for maximin as follows:
\begin{align*}
    \CC &= \overrightarrow{\UU} \cup \{w,x,c,a\}\\
    \succ_t &= c \succ w \succ a \succ \overrightarrow{\UU} \succ x\\
    \succ_B &= c \succ a \succ \overrightarrow{\UU} \succ w \succ x
\end{align*}
Here $w$ is the winner in the unbribed preference profile, and $c$ is the winner in the bribed preference profile. Let $A$ be the set $\UU \cup \{w, x\}$.

We now construct a set $\VV$ of $2m+13t+2$ voters having the following types of preference orders in \PP. There are $m$ voters ($\forall i \in [m]$), each of types (i)-1 and (i)-2;

(i)-1: $x \succ w \succ \overrightarrow{S_i} \succ a \succ c \succ \overrightarrow{\UU \setminus S_i}$

(i)-2: $\overleftarrow{U \setminus S_i} \succ c \succ a \succ \overleftarrow{S_i} \succ w \succ x$
        
\noindent There are $t-1$ voters each, of types (ii)-1 and (ii)-2:

(ii)-1: $\overrightarrow{\UU} \succ w \succ c \succ x \succ a$

(ii)-2: $a \succ x \succ c \succ w \succ \overleftarrow{\UU}$
        
\noindent There are $t$ voters each, of types (iii)-1 and (iii)-2:

(iii)-1: $a \succ c \succ \overrightarrow{\UU} \succ x \succ w$

(iii)-2: $w \succ x \succ \overleftarrow{\UU} \succ a \succ c$
        
\noindent There are $3t+1$ voters each, of types (iv)-1 and (iv)-2:

(iv)-1: $c \succ \overrightarrow{\UU} \succ w \succ x \succ a$

(iv)-2: $a \succ c \succ x \succ w \succ \overleftarrow{\UU}$
        
\noindent Let us define $3t+2$ votes of type (v) as follows. Consider a fixed ordering $\overrightarrow{A_1}$ on the set $A$. To get $\overrightarrow{A_i}$,  $\forall i \in [3t+2]$, $\overrightarrow{A_1}$ is rotated clockwise by $i-1$ candidates:

(v)-i: $\overrightarrow{A_i} \succ a \succ c$

Considering the above preference orders, we have:
\begin{align*}
vs(a,c) &= 5t + 2\\
vs(a,a_i) &= -3t-2, \forall a_i \in A\\
vs(x,c) &= -3t\\
vs(u_i,c) &= -3t, \forall u_i \in \UU\\
vs(w,c) &= -t-2\\
vs(a_i,a_j) &\geq -3t, \forall a_i, a_j \in A
\end{align*}
Here, $\VV_b$ is equal to the set (i)-1 of $m$ voters. For each $i \in [m]$, we have a bribed voter $v_i \in \VV_b$ with original vote
$$ \succ_i^\PP = x \succ w \succ S_i \succ a \succ c \succ \overrightarrow{\UU \setminus S_i}$$
and bribed vote
$$ \succ_i^\QQ = x \succ c \succ w \succ S_i \succ a \succ \overrightarrow{\UU \setminus S_i}$$

Now we show that the given instance of \SBIS for maximin is a \NO instance if and only if the instance of \XTC is a \YES instance.

    $\Longrightarrow:$ Assume the instance of \SBIS for maximin is a \NO instance, implying that a bad candidate ($x$) is the winner in some preference profile \RR where a subset of the bribed voters vote according to the briber's suggested preference orders. We now show that this preference profile corresponds to a solution of the \XTC problem. Let $Y \in \VV_b$ be the set of voters who vote according to $\succ_i^\QQ$ ($\forall v_i \in Y$), so for all $v_i \in Y, \succ_i^\RR = \succ_i^\QQ$ and for all $v_i \notin Y, \succ_i^\RR = \succ_i^\QQ$. We show that for $x$ to win, $|Y| = t$. 
    Clearly, $|Y| \leq t$, else the maximin scores of $c$ and $x$ become equal, and in tie-breaking, $c$ wins.
    Also, $|Y| \geq t$, else $w$ wins. Now, for $x$ to win, the preference orders of the voters in $Y$ must correspond to an exact cover of \UU, otherwise $\exists u_i \in \UU$, such that her maximin score remains equal to $x$ and in tie-breaking, $u_i$ wins against $x$. Hence, \XTC is a \YES instance.
    
    $\Longleftarrow:$ Assume the instance of \XTC is a \YES instance. Now, to show that \SBIS for maximin is a \NO instance, it suffices to show that a bad candidate ($x$), wins. Let $X \subseteq \SS$ be the set of triplets that form an exact cover of $\UU$. Clearly $|X| = t$. Now, if for every $S_i \in X$, we use the preference order $\succ_i^\QQ$ instead of $\succ_i^\PP$, $vs(w,c)$ becomes $-3t-2$, $vs(u_i,c)$ becomes $-3t-2, \forall i \in [3t]$, and $vs(a,c)$ becomes $3t+2$. Therefore the final maximin score of $x(=-3t)$  is greater than those of $c$, $a$ and $u_i, \forall i \in [3t]$. Therefore, $x$ becomes the maximin winner. Hence, there exists a possible voting profile, where a bad candidate (i.e. $x$) wins. Hence \SBIS for maximin is a NO instance. 
\end{proof}


Next we show the parameterized hardness results for \SBIS and \SSB.
\section{Parameterized Complexity Results}\label{sec:param}
We observe that for each of Copeland, Borda and maximin, \SBIS is fixed parameter tractable when parameterized by the number of shifts.

\begin{theorem}\label{param is safe shifts}
For all anonymous and efficient voting rules, \SBIS parameterized by the number of shifts is fixed parameter tractable.
\end{theorem}

\begin{proof}
Let $(\CC,\PP,c,\succ_B,\QQ)$ be an instance of \SBIS for any anonymous and efficient voting rule. We know that the number of bribed voters ($n_b$) is at most the total number of shifts ($t$). Also the maximum possible shift for any voter is at most $t-n_b+1$. By enumerating over all possible shift actions (including partial), we need to check $\OO(t^t)$ possible shift actions and for each shift action winner determination is $n^{\OO(1)}$. So, this problem is fixed parameter tractable with complexity $\OO(t^t)$poly$(m, n)$.  
\end{proof}

We next see that for Copeland, Borda and maximin, \SSB is in \XP when parameterized by the number of shifts. But for Copeland, we have an added result in \Cref{Coro 2}, that it is \WOH with number of shifts as the parameter. Do note that the methods in \cite{xia2014fixed} cannot be used to get FPT results for \SSB because it is a harder version of the generalized bribery problem.
\begin{theorem}\label{param safe shifts}
For all anonymous and efficient voting rules, \SSB parameterized by the number of shifts is in \XP.
\end{theorem}

\begin{proof}
Let $(\CC, \PP, c, \succ_B, \Pi, b)$ be an instance of \SSB for any anonymous and efficient voting rule. The maximum possible shift for any of the $n$ voters is at most $t$ and at least $0$. Since the total number of possible bribed profiles are $\binom{n+t-1}{t}$ which is $\OO((n+t)^t)$, we can check for each such bribed profile whether it is successful, safe and lies within the budget, all of which are fixed-parameter tractable problems (with respect to the number of shifts).
\end{proof}

To show W[k]-hardness, it is enough to give a parameterized reduction from a known hard problem. Our parameterized hardness proofs for \SBIS, considering the number of bribed voters as a parameter rely on reduction from the \WOH problem \MCI.
    
    \begin{definition}[\MCI]Given a graph $\GG = (V,E)$ where each vertex has one of $h$ colours; compute whether
    there are $h$ vertices of pairwise-distinct colours such that no two of them are connected by an edge.
	\end{definition}
	
	The above problem can be proved to be \WOH, by reducing it from a variant of the \MCC problem \cite{MATHIESON2012179}.
	
In the following theorem, we prove that \SBIS for Borda is \coWOH when parameterized by the number of bribed voters. The basic structure of this proof is inspired by \cite{bredereck2016complexity}.

\begin{theorem}\label{borda param voters}
For Borda, \SBIS is \coWOH, when parameterized by the number of bribed voters.
\end{theorem}

\begin{proof}
Let $(\CC,\PP,c,\succ_B,\QQ)$ be an instance of \SBIS for Borda. We give a parameterized reduction from the \WOH \MCI problem. Given a graph $\GG = (V (\GG), E(\GG))$ where each vertex has one of $h$ colours. Let $(\GG, h)$ be our input instance. Without loss of generality, we assume that the number of vertices of each colour is the same and that there are no edges
between vertices of the same colour. We write $V(\GG)$ to denote the set of $\GG$’s vertices, and $E(\GG)$ to
denote the set of $\GG$’s edges. Further, for every colour $i \in [h]$, we write $V^{(i)} = \{\vvv^{(i)}_1, \dots , \vvv^{(i)}_q\}$ to
denote the set of vertices of colour $i$. For each vertex $\vvv$, we write $E(\vvv)$ to denote the set of edges
incident to $\vvv$. For each vertex $\vvv$, we write $\delta(\vvv)$ to denote its degree, i.e., $\delta(\vvv) = |E(\vvv)|$ and we let $\Delta = max_{u \in V(\GG)} 
\delta(u)$ be the highest degree of a vertex $\GG$.
We form an instance of \SBIS for Borda as follows. We let the candidate set be $\CC = \{c,x\} \cup V(\GG) \cup E(\GG) \cup F(\GG) \cup D' \cup D''$, where $F(\GG)$, $D'$, and $D''$ are sets of special dummy candidates. We let $D'$ and $D''$ have a cardinality of $B+1$ each, where $B = h(q + (q - 1)\Delta)$. For each vertex $\vvv$, we let $F(\vvv)$ be a set of $\Delta - \delta(\vvv)$ dummy candidates, and we let $F(\GG) = \cup_{\vvv \in V(\GG)} F(\vvv)$. We set $F(-i) = \cup_{\vvv \in V(i'),i' \neq i} F(\vvv)$. Let $w \in V(\GG) \cup E(\GG)$ be the winner in $\PP$. For each vertex $\vvv$, we define the partial preference order $\overrightarrow{S(\vvv)}$:
$$\overrightarrow{S(\vvv)}: \vvv \succ \overrightarrow{E(\vvv)} \succ \overrightarrow{F(\vvv)}$$
For each colour $i$, we define $\overrightarrow{R(i)}$ to be a partial preference order that ranks first all members of $D'$, then all vertex candidates of colours other than $i$, then all edge candidates corresponding to edges that are not incident to a vertex of colour $i$, then all dummy vertices from $F(-i)$, and finally all candidates from $D''$.
Let the briber's preference order, $\succ_B$, be: $c \succ \overrightarrow{\CC \setminus \{c,x,w\}} \succ w \succ x$, and
let the tie-breaking rule, $\succ_t$, be $c \succ w \succ \overrightarrow{\CC \setminus \{c,w,x\}} \succ x$.
We therefore let $x$ be the only bad candidate.

Let $\VV_b$ be the set of bribed voters. We define two sets of $h$ voters each, $\VV_{b_1}$ and $\VV_{b_2}$ such that $\VV_b = \VV_{b_1} \uplus \VV_{b_2}$. For each $i \in [h]$, we have a voter $v_{i_1} \in \VV_{b_1}$, whose preferences in $\PP$ and $\QQ$ are:
$$\succ_{i_1}^{\PP} = \overrightarrow{S(\vvv_1^{(i)})} \succ \overrightarrow{S(\vvv_2^{(i)})} \succ \dots \succ \overrightarrow{S(\vvv_q^{(i)})} \succ c \succ x \succ \overrightarrow{R(i)}$$
$$\succ_{i_1}^{\QQ} = c \succ \overrightarrow{S(\vvv_1^{(i)})} \succ \overrightarrow{S(\vvv_2^{(i)})} \succ  \dots \succ \overrightarrow{S(\vvv_q^{(i)})} \succ x \succ \overrightarrow{R(i)}$$
Similarly for each $i \in [h]$ we have a voter $v_{i_2} \in \VV_{b_2}$, whose preferences in $\PP$ and $\QQ$ are:
$$\succ_{i_2}^{\PP} = \overleftarrow{S(\vvv_q^{(i)})} \succ \overleftarrow{S(\vvv_{q-1}^{(i)})} \succ  \dots \succ \overleftarrow{S(\vvv_1^{(i)})} \succ c \succ x \succ \overrightarrow{R(i)}$$
$$\succ_{i_2}^{\QQ} = c \succ \overleftarrow{S(\vvv_q^{(i)})} \succ \overleftarrow{S(\vvv_{q-1}^{(i)})} \succ  \dots \succ \overleftarrow{S(\vvv_1^{(i)})} \succ x \succ \overrightarrow{R(i)}$$

Let $\VV_u$ be the voters who were not bribed. They are of the following types. There are $h$ voters ($\forall i \in [h]$), each of types (i)-1 and (i)-2;

(i)-1: $\overleftarrow{R(i)} \succ x \succ c \succ \overleftarrow{S(\vvv_q^{(i)})} \succ \dots \succ \overleftarrow{S(\vvv_2^{(i)})} \succ \overleftarrow{S(\vvv_{1}^{(i)})}$

(i)-2: $\overleftarrow{R(i)} \succ x \succ c \succ \overrightarrow{S(\vvv_1^{(i)})} \succ \dots \succ \overrightarrow{S(\vvv_{q-1}^{(i)})} \succ \overrightarrow{S(\vvv_q^{(i)})}$

\noindent There is 1 voter, each of type (ii)-1 and (ii)-2;

(ii)-1: $\overrightarrow{F(\GG)} \succ \overrightarrow{V(\GG)} \succ x \succ \overrightarrow{E(\GG)} \succ \overrightarrow{D'} \succ c \succ \overrightarrow{D''}$

(ii)-2: Reverse (ii)-1, and then shift $c$ to the right by $B-1$ places and shift $V(\GG) \cup \{x\} \cup E(\GG)$ to the left by $1$ place.

Let $L$ be the score of $c$ prior to executing any shift actions. Simple calculations show that each candidate in $V (\GG) \cup \{x\} \cup E(\GG)$ has score $L + B + 1$, and each candidate in $F(\GG) \cup D' \cup D''$ has score at most $L + B$. Now, it is easy to show that \SBIS for Borda is a \NO instance if and only if \MCI is a \YES instance. 

$\Longrightarrow$: Here, we already know that the instance of \SBIS for Borda is a \NO instance, so $x$ is a winner when a specific subset of bribed voters perform certain shifts (an unsafe shift action). A shift action of total number of shifts $B$ gives $c$ a score of $L + B$. No more shifts are possible in an unsafe shift action, otherwise $c$ would win instead of $x$. Also, for the shift action to be unsafe, it has to cause all candidates in $V(\GG) \cup E(\GG)$ to lose a point. We claim that an unsafe shift bribery has to use exactly $\nicefrac{B}{h} = (q + (q - 1)\Delta)$ unit shifts for every pair of voters with preferences $\succ_{i_1}^{\PP}$, $\succ_{i_2}^{\PP}$. The vertices of colour $i$ are only present in votes cast by the pair of voters $v_{i_1}$ and $v_{i_2}$, which indicates that $c$ must pass all $q$ vertex candidates of colour $i$ using these votes. Now, the total number of shifts allotted to $c$ in this pair of votes cannot be less than $\nicefrac{B}{h}$ or else some vertex candidate will not be crossed and it will defeat $x$. This also proves the fact that the number of shifts for such a pair of voters cannot be greater than $\nicefrac{B}{h}$ as then there will exist some other colour $j \in [h]$ for which the total number of available shifts will be less than $\nicefrac{B}{h}$ and some vertex candidate will defeat $x$.

Further, we can assume that for each colour $i$ there is a vertex $\vvv^{(i)}_{s_i} \in V^{(i)}$ such that in $\succ^{\PP}_{i_1}$ candidate $c$ is shifted to be right in front of $\vvv^{(i)}_{s_{i+1}}$ and in $\succ^{\PP}_{i_2}$
candidate $c$ is shifted to be before $\vvv^{(i)}_{s_i}$. We call such a vertex $\vvv^{(i)}_{s_i}$ selected. If for a given pair of voters $v_{i_1}$, $v_{i_2}$ neither of the vertices from $V^{(i)}$ were selected, then there would be some vertex candidate in $V^{(i)}$ that $c$ does not pass. If for some pair of voters $v_{i_1}$, $v_{i_2}$ vertex $\vvv^{(i)}_{s_i}$
is selected, then in this pair of votes $c$ does not pass the edge candidates from $E(\vvv^{(i)}_{s_i})$.

However, this ensures that in an unsafe shift action the selected vertices form an independent set of $\GG$. Else, if two vertices $\vvv^{(i)}_{s_i}$ and $\vvv^{(j)}_{s_j}$ were selected, $i \neq j$, and if there were an edge $e$ connecting them, then $c$ would not pass the candidate $e$ in either of the pairs of votes cast by $v_{i_1}$, $v_{i_2}$ or $v_{j_1}$, $v_{j_2}$. Since these are the only votes where $c$ can pass $e$ without exceeding the number of shifts ($\nicefrac{B}{h}$), in this case $e$ would have $L + B + 1$ points, and would defeat $x$ in tie-breaking. Therefore \MCI is a \YES instance.

$\Longleftarrow$: Assume the instance of \MCI is a \YES instance. To show that \SBIS for Borda is a \NO instance, it is equivalent to show that a bad candidate (here $x$), wins. For $x$ to win, she must defeat the rest of the candidates.
Let us fix a multicoloured independent set for $\GG$ and, for each
colour $i \in [h]$, let $\vvv^{(i)}_{s_i}$ be the vertex of colour $i$ from this set. First of all we note that the total number of shifts $\leq B$ or else $c$ defeats $x$. Also, the total number of shifts is $\geq B$ or else there exists some vertex candidate which defeats $x$ in tie-breaking. Hence the total number of shifts is exactly $B$. Now, for each pair of voters $\vvv_{i_1}$ and $\vvv_{i_2}$, we shift $c$ in their votes (corresponding to $\PP$), such that in $\succ_{i_1}^{\PP}$ she ends up right before $\vvv^{(i)}_{s_{i+1}}$ (or $c$ does not move if $s_i = q$), and in $\succ_{i_2}^{\PP}$ she ends up right before $\vvv^{(i)}_{s_i}$. This way, $c$ passes every vertex candidate from $V^{(i)}$ and every edge candidate from $(\cup_{t\in[q]} E(\vvv^{(i)}_t)) \setminus E(v^{(i)}_{s_i})$. This shift action uses $\nicefrac{B}{h}$ shifts for every pair of voters, so, in total, uses exactly $B$ shifts. Furthermore, clearly, it ensures that $c$ passes every vertex candidate so each of them has score $L + B$. Finally, since we chose vertices from an independent set, every
edge candidate also has score at most $L + B$: If $c$ does not pass some edge $e$ between vertices of colours $i$ and $j$ for a pair of voters $v_{i_1}$ and $v_{i_2}$, then $c$ certainly passes $e$ in the pair of votes cast by $v_{j_1}$ and $v_{j_2}$, because $\vvv^{(i)}_{s_i}$ and $\vvv^{(j)}_{s_j}$ are not adjacent. Therefore, \SBIS for Borda is a \NO instance.
\end{proof}

Next, we have a similar result for the Copeland voting rule parameterized by the number of bribed voters.
\begin{theorem}\label{copeland param voters}
For Copeland, \SBIS is \coWOH, when parameterized by the number of bribed voters.
\end{theorem}

\begin{proof}
Let $(\CC,\PP,c,\succ_B,\QQ)$ be an instance of \SBIS for Copeland. We give a parameterized reduction from the \WOH \MCI problem. Let $(\GG, h)$ be our input instance. Without loss of generality, we assume that the number of vertices of each colour is the same and that there are no edges
between vertices of the same colour. The notation for the vertices, edges, coloured vertices, incident edges, and degree is as defined in \Cref{borda param voters}.
We form an instance of \SBIS for Copeland as follows. We let the candidate set be:
$$ \CC = \{c,x\} \cup V(\GG) \cup E(\GG) \cup F(\GG) \cup H \cup J \cup D $$
Let $\vvv_j^{(i)}$ be the $j^{\text{th}}$ ranked vertex of the $i^{\text{th}}$ colour. For each vertex $v$, we let $F(\vvv)$ be a set of $\Delta - \delta(\vvv)$ filler candidates, and we let $F(\GG) = \cup_{\vvv \in V(\GG)} F(\vvv)$. We set $F(-i) = \cup_{\vvv \in V(i'), i' \neq i} F(\vvv)$. For each vertex $v_k^{(i)}$, we define two dummy candidates $h_k^{(i)}$ and $j_k^{(i)}$. We define $H(i) = \cup_{k \in [q]} \{h_k^{(i)}\}$ and $J(i) = \cup_{k \in [q]} \{j_k^{(i)}\}$. Also, we define $H(-i) = \cup_{v \in V(i'), i' \neq i} H(i)$ and $J(-i) = \cup_{v \in V(i'), i' \neq i} J(i)$. Also, $H = \cup_{i \in [h]} H(i)$ and $J = \cup_{i \in [h]} J(i)$ and $D$ is a set of $B$ dummy candidates, where $B = hq + h(q-1)\Delta$.
For each colour $i$, we define $\overrightarrow{R(i)}$ to be a partial preference order that ranks first all members of $H(-i)$, then all members of $J(-i)$, then all vertex candidates of colours other than $i$, then all edge candidates corresponding to edges that are not incident to a vertex of colour $i$, then all vertices from $F(-i)$, and finally, all vertices from $D$. Let $w \in V(\GG) \cup E(\GG)$ be the winner in $\PP$.
Let the briber's preference order, $\succ_B$, be:
$$c \succ \overrightarrow{\CC \setminus \{w, c, x\}} \succ w \succ x$$
Let the tie-breaking rule, $\succ_t$, be:
$$c \succ w \succ \overrightarrow{\CC \setminus \{w, c, x\}} \succ x$$

This makes $x$ the only bad candidate. We define $\overrightarrow{S_k^{(i)}} = v_k^{(i)} \succ h_{k-1}^{(i)} \succ \overrightarrow{E_k^{(i)}} \succ \overrightarrow{F_k^{(i)}} \succ j_{k-1}^{(i)}$ and $\overrightarrow{T_k^{(i)}} = j_k^{(i)} \succ \overleftarrow{F_k^{(i)}} \succ \overleftarrow{E_k^{(i)}} \succ h_k^{(i)} \succ v_k^{(i)}$.
Let $\VV_b$ be the set of bribed voters. We define two disjoint sets of $h$ voters each, $\VV_{b_1}$ and $\VV_{b_2}$ such that $\VV_b = \VV_{b_1} \uplus \VV_{b_2}$. For each $i \in [h]$, we have a voter $v_{i_1} \in \VV_{b_1}$, whose preferences in $\PP$ and $\QQ$
are as follows:
$$ \succ_{i_1}^{\PP} =  \overrightarrow{S_1^{(i)}} \succ \ldots \succ \overrightarrow{S_q^{(i)}} \succ c \succ x \succ h_q^{(i)} \succ j_q^{(i)} \succ \overrightarrow{R(i)}$$
$$ \succ_{i_1}^{\QQ} =  c \succ \overrightarrow{S_1^{(i)}} \succ \ldots \succ \overrightarrow{S_q^{(i)}} \succ x \succ h_q^{(i)} \succ j_q^{(i)} \succ \overrightarrow{R(i)}$$

For each $i \in [h]$, we have a voter $v_{i_2} \in \VV_{b_2}$, whose preferences in $\PP$ and $\QQ$ are as follows:
$$ \succ_{i_2}^{\PP} = \overrightarrow{T_q^{(i)}} \succ \ldots \succ \overrightarrow{T_1^{(i)}} \succ c \succ x \succ j_0^{(i)} \succ h_0^{(i)} \succ \overrightarrow{R(i)}$$
$$ \succ_{i_2}^{\QQ} = c \succ \overrightarrow{T_q^{(i)}} \succ \ldots \succ \overrightarrow{T_1^{(i)}} \succ x \succ j_0^{(i)} \succ h_0^{(i)} \succ \overrightarrow{R(i)}$$
Let $\VV_u$ be the voters who were not bribed. They are as follows:
There are $[h]$ voters ($\forall i \in [h]$), each of type (i)-1 and (i)-2:

(i)-1: $\overleftarrow{R(i)} \succ j_q^{(i)} \succ h_q^{(i)} \succ x \succ c \succ \overleftarrow{S_q^{(i)}} \succ \ldots \succ \overleftarrow{S_1^{(i)}}$

(i)-2: $\overleftarrow{R(i)} \succ h_0^{(i)} \succ j_0^{(i)} \succ x \succ c \succ \overleftarrow{T_1^{(i)}} \succ \ldots \succ \overleftarrow{T_q^{(i)}}$

\noindent Consider a fixed ordering $\overrightarrow{\DD_1}$ on the set $D$. To get $\overrightarrow{\DD_k}$, we rotate $\overrightarrow{\DD_1}$ clockwise by $k-1$ candidates. We define $(|D|-1)/2$ voters of type (ii)-1 and (ii)-2 each as follows:

(ii)-(2k-1): $\overrightarrow{F} \succ \overrightarrow{E} \succ \overrightarrow{V} \succ x \succ \overrightarrow{H} \succ \overrightarrow{J} \succ \overrightarrow{\DD_{2k-1}} \succ c$

(ii)-(2k): $c \succ \overrightarrow{\DD_{2k}} \succ \overleftarrow{J} \succ \overleftarrow{H} \succ x \succ \overleftarrow{V} \succ \overleftarrow{E} \succ \overleftarrow{F}$

\noindent Consider an ordering $\overrightarrow{\EE_1}$ over the set $E \cup V \cup x$. We define 1 voter of type (iii), as  follows:

(iii): $\overrightarrow{\EE_1} \succ \overrightarrow{F} \succ \overrightarrow{H} \succ \overrightarrow{J} \succ \overrightarrow{\DD_1} \succ c$

\noindent To get $\overrightarrow{\EE_k}$, we rotate $\overrightarrow{\DD_1}$ clockwise by $k-1$ candidates. We define $(|E|+|V|)/2$ voters of type (iv)-(2k) and (iv)-(2k+1) each as follows:

(iv)-(2k): $c \succ \overleftarrow{\DD_1} \succ \overleftarrow{J} \succ \overleftarrow{H} \succ \overleftarrow{F} \succ \overrightarrow{\EE_{2k}}$

(iv)-(2k+1): $\overrightarrow{\EE_{2k+1}} \succ \overrightarrow{F} \succ \overrightarrow{H} \succ \overrightarrow{J} \succ \overrightarrow{\DD_1} \succ c$

\noindent $\forall i \in [h]$ and $\forall k \in [q+1]$, we have one vote each of the following two types:

(v)-1: $h_{k-1}^{(i)} \succ j_{k-1}^{(i)} \succ c \succ \overrightarrow{\CC \setminus \{h_{k-1}^{(i)} \cup j_{k-1}^{(i)} \cup c\}}$

(v)-2: $\overleftarrow{\CC \setminus \{h_{k-1}^{(i)} \cup j_{k-1}^{(i)} \cup c\}} \succ j_{k-1}^{(i)} \succ h_{k-1}^{(i)} \succ c$

Due to the voters in $\VV_u$, the Copeland scores of the candidates are as follows: For each candidate $a \in V(\GG) \cup \{x\} \cup E(\GG)$, $s(a) = B+1$ and $s(c)=0$. For each candidate $b \in D \cup H \cup J$, $s(b) < B+1$. We can also see that the important pairwise vote scores with respect to $c$ are $vs(E,c) = vs(F,c) = vs(V,c) = 1$ and $vs(H,c) = vs(J,c) = 3$.

Now, it is easy to show that the given instance of \SBIS for Copeland is a \NO instance if and only if the instance of \MCI is a \YES instance.

$\Longrightarrow$: Here, we already know that the instance of \SBIS for Copeland is a \NO instance, so $x$ is a winner when a specific subset of bribed voters perform certain shifts (an unsafe shift action). A shift action of total number of shifts $B+2h(q-1)$ gives $c$ a score of $B$. If the total number of shifts were less than this, then for a bribed vote pair corresponding to some colour would have total number of shifts less than $\nicefrac{B}{h}+2(q-1)$, implying that some $\vvv \in V(\GG)$ would not be crossed by $c$ and would win. Hence, the bribery would be safe, which is a contradiction. Also, for the shift action to be unsafe, it has to cause all candidates in $V(\GG) \cup E(\GG)$ to lose a point. We claim that an unsafe shift action must use no more than $\nicefrac{B}{h} + 2(q - 1))$ unit shifts for every pair of voters with preferences $\succ_{i_1}^{\PP}$, $\succ_{i_2}^{\PP}$. If this were not true, then $c$ would defeat some candidate from $H \cup J$. Now, if $c$ were to defeat every edge candidate, then $c$ would win, rendering the election safe. In the alternative case, if an edge candidate $e \in E(\GG)$ were to beat $c$ in pairwise election, then $e$ would beat $x$, hence the election would be safe. These are contradictions. 

Further, we can assume that for each colour $i$ there is a vertex $\vvv^{(i)}_{s_i} \in V^{(i)}$ such that in $\succ^{\PP}_{i_1}$ candidate $c$ is shifted to end up right before $\vvv^{(i)}_{s_{i+1}}$ and in $\succ^{\PP}_{i_2}$
candidate $c$ is shifted to end up right before $\vvv^{(i)}_{s_i}$. We call such a vertex $\vvv^{(i)}_{s_i}$ `selected'. If for a given pair of voters $v_{i_1}$, $v_{i_2}$ neither of the vertices from $V^{(i)}$ were selected, then there would be some vertex candidate in $V^{(i)}$ that $c$ does not pass. If for some pair of voters $v_{i_1}$, $v_{i_2}$ vertex $\vvv^{(i)}_{s_i}$ is selected, then in this pair of votes $c$ does not pass the edge candidates from $E(\vvv^{(i)}_{s_i})$.

However, this ensures that in an unsafe shift action the selected vertices form an independent set of $\GG$. Else, if two vertices $\vvv^{(i)}_{s_i}$ and $\vvv^{(j)}_{s_j}$ were selected, $i \neq j$, and if there were an edge $e$ connecting them, then $c$ would not cross the candidate $e$ in either of the pairs of votes cast by $v_{i_1}$, $v_{i_2}$ or $v_{j_1}$, $v_{j_2}$. (These are the only votes where $c$ can cross $e$ without crossing a candidate in $H \cup J$, and thereby making the election safe). This would imply $e$ having $B + 1$ points, and hence defeating $x$ in tie-breaking, rendering the election unsafe. This again is a contradiction. Therefore, \MCI is a \YES instance.

$\Longleftarrow$: Assume the instance of \MCI is a \YES instance. Showing that \SBIS for Copeland is a \NO instance is equivalent to showing that a bad candidate (here $x$), can win. For $x$ to win, she must defeat all the other candidates. Consider that $\GG$ has a multicoloured independent set $I$ and, for each colour $i \in [h]$, let $\vvv^{(i)}_{s_i}$ be the vertex of colour $i$ from this set. It can be shown that $x$ wins using $B+2h(q-1)$ shifts. For each pair of voters $v_{i_1}$ and $v_{i_2}$, we shift $c$ in their votes (corresponding to $\PP$), such that in $\succ_{i_1}^{\PP}$ she ends up right before $\vvv^{(i)}_{s_{i}+1}$ (or $c$ ends up right before $h_q^{(i)}$, if $s_i = q$), and in $\succ_{i_2}^{\PP}$ she ends up right before $\vvv^{(i)}_{s_i}$. This way, $c$ passes every vertex candidate from $V^{(i)}$ and every edge candidate from $(\cup_{t\in[q]} E(\vvv^{(i)}_t)) \setminus E(\vvv^{(i)}_{s_i})$. This shift action uses $\nicefrac{B}{h}+2(q-1)$ shifts for every pair of voters, so, in total, uses exactly $B$ shifts. Furthermore, it ensures that $c$ passes every vertex candidate
so each of them has a final score of $B$. Finally, since we chose vertices from an independent set, every
edge candidate also has score at most $B$: If $c$ does not pass some edge $e$ between vertices of
colours $i$ and $j$ for a pair of voters $v_{i_1}$ and $v_{i_2}$, then $c$ certainly passes $e$ in the pair of votes cast by $v_{j_1}$ and $v_{j_2}$, because $\vvv^{(i)}_{s_i}$ and $\vvv^{(j)}_{s_j}$ are not adjacent. Therefore, \SBIS for Copeland is a \NO instance.
\end{proof}

\section{Conclusion}
In this paper, we propose and study a nuanced notion of bribery which we call safe bribery. We observe that the computational complexity of safe bribery, for both \$bribery and shift bribery, matches with the classical bribery problem for common voting rules. Hence, safety during bribery can often be achieved without incurring much additional computational overhead. However, we obtained some interesting results for $k$-approval, $k$-veto and simplified Bucklin, which were in \Pb for the shift bribery problems but hard for \$bribery problems. Our work is a natural extension of the bribery problem and can be further studied with respect to approximation algorithms and multi-winner voting rules.

\bibliographystyle{plainurl} 
\bibliography{lipics-v2019-sample-article}

\end{document}